\documentclass[11pt,a4paper]{article}
\usepackage{amssymb}
\usepackage{amsmath}
\usepackage{amsthm}
\usepackage{color,xcolor}
\usepackage{hyperref}
\usepackage{latexsym}
\usepackage[dvips]{epsfig}
\usepackage{mathrsfs}
\usepackage{bm}
\usepackage[utf8]{inputenc}
\usepackage[T1]{fontenc}

\theoremstyle{plain}
\newtheorem{proposition}{Proposition}

\newtheorem{theorem}{Theorem}
\newtheorem{assumption}{Assumption}

\newtheorem{corollary}{Corollary}

\newtheorem{remark}{Remark}

\font\SYM=msbm10

\font\tenscr=rsfs10 scaled1100
\font\sevenscr=rsfs7 
\font\fivescr=rsfs5 
\skewchar\tenscr='177
\skewchar\sevenscr='177
\skewchar\fivescr='177
\newfam\scrfam
\textfont\scrfam=\tenscr
\scriptfont\scrfam=\sevenscr
\scriptscriptfont\scrfam=\fivescr

\newcommand{\Scri}{{\fam\scrfam I}}
\def\scri{{\fam\scrfam I}}

\def\Hubble{H}

\def\Ncf{n}
\def\NCF{\hat{n}}
\def\Acf{a}
\def\ACF{\hat{a}}

\newcommand{\tensor}[3]{_{#1\phantom{#2}#3}^{\phantom{#1}#2}}
\newcommand{\half}{\frac{1}{2}}
\newcommand{\third}{\frac{1}{3}}

\newcounter{mnotecount}[section]

\renewcommand{\themnotecount}{\thesection.\arabic{mnotecount}}

\newcommand{\mnote}[1]
{\protect{\stepcounter{mnotecount}}$^{\mbox{\footnotesize
$
\bullet$\themnotecount}}$ \marginpar{
\raggedright\tiny\em
$\!\!\!\!\!\!\,\bullet$\themnotecount: #1} }

\def\be{\begin{equation}}
\def\ee{\end{equation}}
\def\bea{\begin{eqnarray}}
\def\eea{\end{eqnarray}}

\def\s{\sigma}

\newcommand{\smfrac}[2]{{\textstyle{#1\over#2}}}

\def\half{\smfrac{1}{2}}
\def\third{\smfrac{1}{3}}

\def\a{\alpha}
\def\e{\epsilon}
\def\g{\gamma}
\def\b{\beta}
\def\c{\gamma}
\def\w{\omega}

\newcommand{\R}{\mathbb{R}}

\begin{document}

\title{\textbf{Asymptotic structure and stability of spatially homogeneous space-times with a positive cosmological constant}}
\author{Christian L\"ubbe$^{\flat}\footnote{email: christian.luebbe@epfl.ch}$ and Filipe C. Mena$^\natural$$^{\ddagger}\footnote{email: filipecmena@tecnico.ulisboa.pt}$\\\\
$^\flat$ Ecole Polytechnique Fédérale de Lausanne (EPFL), Lausanne, Switzerland\\
$^\natural$ Centro de An\'alise Matem\'atica, Geometria e Sistemas Din\^amicos,\\
Instituto Superior T\'ecnico, Univ. Lisboa, Av. Rovisco Pais 1, 1049-001 Lisboa, Portugal\\
$^{\ddagger}$ Centro de Matem\'atica, Universidade do Minho, 4710-057 Braga, Portugal}

\maketitle

\begin{abstract}
We investigate the future asymptotics of spatially homogeneous space-times with a positive cosmological constant by using and further developing geometric conformal methods in General Relativity. For a large class of source fields, including fluids with anisotropic stress,
we prove that the space-times are future asymptotically simple and geometrically conformally regular. We use that result in order to show the global conformal regularity of the Einstein-Maxwell system as well as the Einstein-radiation, Einstein-dust, massless Einstein-Vlasov and particular Einstein-scalar field systems for Bianchi space-times. Taking into account previous results, this implies the future non-linear stability of some of those space-times in the sense that, for small perturbations, the space-times approach locally the de Sitter solution asymptotically in time. This extends some cosmic no-hair theorems to almost spatially homogeneous space-times. However, we find that the conformal Einstein field equations preserve the Bianchi type even at conformal infinity, so the resulting asymptotic space-times have {\em conformal hair}. 

\end{abstract}
Keywords: General Relativity; Conformal Geometry; Non-linear stability; Asymptotics
\newpage
\tableofcontents
\newpage
\section{Introduction}

\subsection{Motivation and background}

Cosmological observations indicate that, on sufficiently large scales, the universe can be modeled by nearly spatially homogeneous space-times. Moreover, supernovae data motivates the inclusion of a positive cosmological constant $\Lambda$ on the Einstein field equations (EFE) to account for the observed cosmological acceleration. These observational facts motivate the mathematical analysis of almost spatially homogeneous cosmological models with $\Lambda>0$. 

In turn, there is a long standing conjecture, called the {\em cosmic no-hair conjecture}, which roughly states that cosmological models with $\Lambda>0$, approach the de Sitter space-time, locally asymptotically in time. Remarkably, a version of this conjecture has been proved by Wald \cite{Wal83} in the case of ever expanding exact spatially homogeneous (Bianchi) space-times satisfying the strong and dominant energy conditions. Extensions of this result to the Einstein-Vlasov system were shown in \cite{Lee,Nun12,Nun13, Tod-2007} and for the Einstein-scalar field system in \cite{Ren04b}. 

This led to the question of the robustness of these results against inhomogeneous deviations from the spatially homogeneous background. The literature about this problem is quite vast, particularly for methods using either exact (therefore, highly symmetric) solutions of the EFE or finite order (usually linear) perturbations about the spatially homogeneous and isotropic space-times of the Friedmann-Lema\^itre-Robertson-Walker (FLRW) family. 

We give a brief overview of some non-linear future stability results with $\Lambda>0$. Given initial data corresponding to an ever expanding space-time, the question is whether small (non-linear) perturbations give maximal globally hyperbolic developments which are future causally geodesically complete and stay close enough to the future of the background solution. 
The non-linear stability of de Sitter space within the class of vacuum space-times was proven by Friedrich using conformal methods \cite{Fri-vacuum}. 
This approach was generalised to stability of de Sitter space for the Einstein-Maxwell-Yang-Mills system \cite{Fri-EMYM} as well as space-times with massive scalar fields \cite{Fri-massive-scalar} or dust \cite{Fri-dust}.

Next consider FLRW backgrounds containing a perfect fluid with a linear equation of state $p=(\gamma-1)\rho$, where $p$ and $\rho$ are the fluid pressure and density. In this case, the future stability has been proven: 

(i) For $1< \gamma< 4/3$ \cite{Rod-Spe,Spe12}, and separately for $\gamma=1$ \cite{Haz-Spe} and $\gamma=4/3$ \cite{Spe13}, using $k=0$ as well as $[T_0,\infty)\times {\mathbb T}^3$ type topology and the harmonic gauge methods of Ringstr\"om \cite{Rin08}; 

(ii) For $\gamma=4/3$ and Gaussian spatial curvature $k=1$, by \cite{LueVal12b} using conformal methods; 

(iii) For $1< \gamma\le 4/3$, using a combination of conformal and harmonic methods \cite{Oli}.

(iv) In \cite{Fri-dust} Friedrich covers $\gamma = 1$ for $M = {\mathbb R} \times S$ with $S = {\mathbb S}^3, {\mathbb T}^3$ or ${\mathbb H}^3_*$.

(v) In \cite{Reula}, it was shown for a wide class of perfect fluid equations of state that in the expanding phase of a FLRW background all small enough non-linear perturbations decay exponentially in time.

A non-linear stability result was also obtained for the FLRW-scalar field system using conformal methods in \cite{Alho-Mena-Kroon}, although with more restrictions on the scalar field and its potential than the far reaching result of Ringstr\"om \cite{Rin08}.  More recently, the future dynamics of near-FLRW solutions to the Einstein-massless-scalar field system with a positive cosmological constant was analysed in \cite{Forno} where a future stability result has been proved. For results within spherical symmetry see \cite{Costa-Mena, Costa-Duarte-Mena}.

The work in \cite{Fri-dust} also shows the stability for Einstein-dust space-times wide class of non-spatially homogeneous background space-times. A specific application to space-times with self-gravitating dust balls is given in \cite{VK-dust-balls}.

\subsection{Objectives and main challenges}

The main goal of this paper is to go beyond the results for FLRW backgrounds in two ways: (i) Use spatially homogeneous but anisotropic backgrounds (Bianchi space-times, in particular) (ii) Analyse larger families of source fields.

Unlike FLRW models, the Bianchi space-times are not conformally flat. It is unclear whether the gauge harmonic methods could be adapted to non-conformally flat cases of large families of matter sources. Nevertheless, some future non-linear stability results have already been proved for the massive Einstein-Vlasov system with a positive cosmological constant on some Bianchi background geometries \cite{Ringstroembook, Nun15, Andreasson-Ring}.

One of the mathematical difficulties in going from FLRW to Bianchi backgrounds is, in fact, the need to control the Weyl tensor. In order to do that, in this paper, we shall sharpen previous estimates for the decay of some of the background quantities.
This allows us to use geometric conformal methods to prove two notions of conformal regularity for a large family of Bianchi spacetimes: i) asymptotic simplicity, and ii) regular solutions to the conformal Einstein field equations (CEFE) up to and beyond null infinity. This will in turn lead to our stability results.

The conformal methods in general relativity (see e.g. \cite{Fri-review, JVKbook, FraLR}) rely on the existence of a regular conformal embedding which can turn a global stability problem into a local problem for which standard PDE results can be applied. If a conformal extension with sufficient regularity exists, then the conformal Weyl tensor must vanish at conformal infinity. In turn, the conformal boundary is uniquely determined by the CEFE once suitable data is prescribed. The stability of the solutions relies on the analysis of the asymptotic structure and the conformal geometry near conformal infinity. 

However, for a generic space-time, such a conformal embedding may not exist. A well known example is the Nariai space-time which is geodesically complete but can not be conformally compactified \cite{Bey09a,Bey09b}. Hence, in order to attack the stability problem, the conformal regularity of a space-time, i.e. the existence of a regular conformal extension to conformal infinity, has to be proven first. 

In \cite{Fri-dust} this problem was avoided by prescribing so-called \emph{Cauchy data at infinity} for the CEFE on a Cauchy surface, that subsequently describes future null infinity, and then evolving this data to the past. However the concern is this class of space-times may be fairly limited within the wider space of solutions to the EFE. 
In this paper, we consider Bianchi space-times with a broad range of matter models satisfying the original conditions of Wald \cite{Wal83}. We will show that for these Bianchi space-times we have sufficient conformal regularity to exploit the approach of Cauchy data at infinity for several matter models.

As we have described, the stability analysis, so far, has been specific to a particular class of matter models. We shall, instead, analyse the conformal regularity\footnote{We will address asymptotic simplicity in section \ref{Sec:Penrose approach} and Bianchi space-times as regular solutions to the CEFE in section \ref{Sec:CEFE}.} of large families of space-times including perfect and non-perfect fluids, the massive and the massless Einstein-Vlasov system, the Einstein-Maxwell system and the massive Einstein-scalar field system.
In order to do so, we need to have adequate decay estimates for the system's variables. Some useful estimates were already derived in \cite{Wal83} and in \cite{Lee}. However, those estimates are not sufficient for our proposes, so part of our task will be to improve them. 

\subsection{Main results and outline of the paper}

Starting from two quite general assumptions for our Bianchi space-times (see sections \ref{sec2.2} and \ref{sec3.4}) we will show quite general decay rates for the background allowing us to treat bigger classes of matter sources. Those estimates are summarised in Proposition \ref{prop-decay1} and Proposition \ref{prop-decay2}. In turn, these will be used to prove the geometric conformal regularity of Bianchi spacetimes in Theorem \ref{prop-future-asymp-simple}, which can be more simply be stated as:
\\\\
{\bf Theorem (Asymptotically Simple Bianchi space-times):} 
\\
{\em Spatially homogeneous Bianchi space-times (not of type IX), initially expanding and satisfying the dominant and strong energy conditions, can be conformally extended to $\scri^+$ and are asymptotically simple, provided the matter anisotropic stresses are zero or decay sufficiently fast.} 
\\\\
As we will show, this result includes various cases of interest such as Vlasov matter, perfect fluids, scalar fields, Maxwell fields and other trace-free fields as summarized in Table \ref{table} at the end of this section.

After proving asymptotic simplicity, we shall study the future non-linear stability of solutions for our background space-times. To do that, we shall first analyse the asymptotic constraints in Proposition \ref{prop-constraints at Scri} and Proposition \ref{thm-constraints}. We then investigate the existence of regular solutions to the CEFE up to conformal infinity in Theorem \ref{space-times-regular-CEFE} and expand them using the approach of prescribing the asymptotic values as Cauchy data at infinity. This latter analysis is hampered by the fact that the formulation of regular CEFE has only been shown for few sources fields, such as radiation or dust fluids, the Einstein-Maxwell system and some massive scalar fields (see Table \ref{table}). For those settings, we shall rely on the results established in \cite{Fri-EMYM, Fri-massive-scalar, Fri-dust, LueVal12b}, which are based on Kato's theorem \cite{Kato}, to deduce the future non-linear stability of Bianchi space-times in suitable Sobolev norms, in Theorem \ref{Main theorem}. Our results in this respect can be roughly summarised as:
\\\\
{\bf Theorem (Conformal Einstein field equations and non-linear stability):}
\\
{\em Spatially homogeneous spacetimes as in the previous theorem and containing either Maxwell fields, dust fluids, radiation fluids or some massive scalar fields, give rise to regular solutions of the conformal Einstein field equations up to and including $\scri^+$. Furthermore, the regular solutions exist beyond the conformal boundary and are non-linearly stable against small perturbations. As a  result, these almost spatially homogeneous space-times locally approach the de Sitter space-time asymptotically in the future.} 
\\\\
The plan of the paper is as follows: Section \ref{setup} contains our preliminary analysis where we revisit Wald's result in terms of the explicit asymptotic decay estimates for the background system of differential equations. 
In Section \ref{Section unphysical variables}, we introduce conformally rescaled variables and use bootstrap arguments to improve the preliminary decay estimates. 
In Section \ref{Section - matter models}, we use the reformulated EFE in order to derive the estimates for different sources fields. 
The asymptotic simplicity and geometric conformal regularity of the space-times is analysed in Section \ref{Section - Conformal regularity}, while Section \ref{Sec:Friedrich approach} is devoted to the non-linear stability.       

We close this section by providing an overview in Table \ref{table} of the various matter models considered and the results we are able to deduce for our class of Bianchi space-times.
\\
\begin{table}[h]
\begin{tabular}{ | l | c | c | c | }
\hline
{\bf Matter models in Bianchi space-times}& \vtop{\hbox{\strut Asymptotic }\hbox{\strut~simplicity }} & \vtop{\hbox{\strut Conformal  }\hbox{\strut~regularity}} & Stability \\
\hline
Maxwell fields & $\checkmark$& $\checkmark$ & $\checkmark$  \\
\hline
Aligned radiation fluids & $\checkmark$& $\checkmark$ & $\checkmark$    \\
\hline
Aligned dust fluids & $\checkmark$& $\checkmark$ & $\checkmark$    \\
\hline
Massive scalar fields with potential \eqref{scalar-field-potential} & $\checkmark$ & $\checkmark$ & $\checkmark$ \\
\hline
Massless Vlasov matter  &$\checkmark$ & $\checkmark$  &      \\
\hline
Massive Vlasov matter & $\checkmark$ &  &     \\
\hline
General trace-free matter &$\checkmark$ &  &     \\
\hline
Aligned perfect fluids with $\gamma \in [\tfrac{2}{3},2]$ & $\checkmark$  & &     \\
\hline
Elastic matter with \eqref{elastic} and $\gamma \ge \frac{4}{3}$& $\checkmark$ &  &  \\
\hline
Viscous fluids with \eqref{viscous-T}, \eqref{viscous-example} and $\gamma \ge \frac{4}{3}$.& $\checkmark$ &  &   \\
\hline
\end{tabular}
\caption{Examples of the matter fields satisfying our main theorems. The study of the cases without a $\checkmark$ is hampered by the fact that the respective CEFE don't exist yet, at least with sufficient generality.}
\label{table}
\end{table}
%
\section{Setup and preliminary analysis}
\label{setup}
In this section, we briefly revise ideas of spatial homogeneity and the energy conditions that will be used later on. We follow the conventions of \cite{vElUgg96}. The metric has signature $(-\, +\, +\, +)$. Space-time indices are denoted by Greek indices $\mu,\nu,\rho,\sigma=0,1,2,3$ and spatial coordinates by Latin indices $i,j,k= 1,2,3$. Orthonormal frame space-time indices are denoted by Latin indices $a,b,c,d=0,1,2,3$ and spatial frame indices by greek indices $\alpha,\beta,\gamma,\delta, \epsilon = 1,2,3$. This should also be clear from the context. We use geometrised units with $c=8\pi G=1$ and take a positive cosmological constant $\Lambda >0$. Hence, the Einstein field equations take the form
\begin{equation}
\label{Einstein equation}
G_{\mu\nu} + \Lambda g_{\mu\nu}= T_{\mu\nu},\quad\quad\Lambda  >0.
\end{equation}
In this article, we will express the cosmological constant in terms of $\lambda = \sqrt{\Lambda/3}$ to ease notation and readability.

We consider a space-time $(M,g)$ with a distinguished time-like direction given by the velocity
vector field ${\bf u}$. 
We use the formalism described in Ehlers \cite{Ehlers} and 
Ellis \cite{Ellis73} and define
a tensor which, at each point, projects into the space orthogonal to
${\bf u}$ by
\begin{equation}
\label{ola1}
h_{\mu\nu}=g_{\mu\nu}+u_\mu u_\nu,
\end{equation}
such that
$h_{~\mu}^{\rho} h_{~\rho}^\nu=h_{~\mu}^\nu,~~~h_{~\mu}^\nu u_\nu=0,~~~h_{~\mu}^\mu=3.$
The covariant derivative of ${\bf u}$ can be decomposed into its irreducible parts as
\begin{equation}
\label{ola3}
\nabla_\nu u_\mu = \Hubble h_{\mu\nu} + \s_{\mu\nu} - \omega_{\mu\nu} - A_\mu u_\nu,
\end{equation}
where
\begin{eqnarray}
\label{ola73}
&&\theta_{\mu\nu} :=\s_{\mu\nu}+\Hubble h_{\mu\nu},~~\sigma_{\mu\nu}:=\nabla_{(\nu} u_{\mu)}-\Hubble h_{\mu\nu}+A_{(\mu}u_{\nu)},~~~\w_{\mu\nu}:=-\nabla_{[\nu}u_{\mu]}-A_{[\mu}u_{\nu]},\nonumber\\
&&A_\mu := u^\nu \nabla_\nu u_{\mu},~~~~~~~
 \Hubble := \third \nabla_\mu u^\mu 
\end{eqnarray}
and $\quad
\s_{\mu\nu}=\s_{(\mu\nu)};~~~\s^{~\mu}_\mu=0;~~~\s_{\mu\nu}u^\nu=0;~~~
\omega_{\mu\nu}=\omega_{[\mu\nu]};~~~\omega_{\mu\nu}u^\nu=0.$
The tensor $\omega_{\mu\nu}$ is interpreted as the vorticity tensor,
$\s_{\mu\nu}$ as the shear, and
$\theta$ the expansion. It is also useful to define
$\s^2=\half\s_{\mu\nu}\s^{\mu\nu}$ and to recall that 
the vector field ${\bf u}$  is hypersurface forming if $\omega_{\mu\nu}=0$. In the context of this article, ${\bf u}$ will be aligned with the normal $\bf n$ of the surfaces of homogeneity.

The stress-energy tensor in General Relativity can decomposed with respect to ${\bf u}$ as
\begin{equation}
\label{ola10}
T_{\mu\nu}=\rho u_\mu u_\nu+2q_{(\mu}u_{\nu)}+ph_{\mu\nu}+\pi_{\mu\nu},
\end{equation}
where 
$q_\mu u^\mu=0;~~~\pi_{\mu\nu}u^\nu=0;~~~\pi^\mu_{~\mu}=0;~~~\pi_{\mu\nu}=\pi_{(\mu\nu)}$.
Interpreting this as an imperfect fluid, the energy density
is represented by $\rho$, the energy flux relative to ${\bf u}$ by $q^\mu$,
the isotropic pressure by $p$ and the anisotropic stress by $\pi_{\mu\nu}$.

We now recall the various energy conditions on the energy-momentum tensor: A condition is said to hold if and only if  the equations hold for any arbitrary causal (i.e. timelike or null) vectors $v^\mu$, $w^\mu$. 
\begin{itemize}
\item \textit{The dominant energy condition (DEC)}: $T_{\mu\nu} v^\mu w^\nu \ge 0$ or equivalently  $T_{\mu\nu} v^\mu $ is a causal vector and $T_{\mu\nu} v^\mu v^\nu \ge 0$.

\item \textit{The strong energy condition (SEC)}: $(T_{\mu\nu}-\half g_{\mu\nu} T ) v^\mu v^\nu \ge 0$.

\item \textit{The weak energy condition (WEC)}: $T_{\mu\nu} v^\mu v^\nu \ge 0$.

\end{itemize} 
The DEC implies the WEC. If the energy-momentum tensor is tracefree, then the SEC and the WEC are equivalent. 
For convenience, we set the initial time to be $t_*$ throughout our analysis and denote the initial value of a quantity with a suffix $*$, e.g. $L_*$. In particular, when dealing with a congruence we will denote a value that is constant along each individual curve by a subscript $*$.

In the following, we will express the time decay rates of different variables by bounding their long term behaviour by various functions. We recall the following notation
\begin{equation}
f(t) = O (g(t)) \quad \iff \quad \exists\, C>0, \, t_0\,\, \textmd{such that} \,\vert f(t) \vert \le C \vert g(t) \vert \quad \forall t>t_0.
\end{equation}
We typically do not specify the constants $C$ or $t_0$ explicitly. We may also write $f(t) - g(t) = O(h(t))$ as
\begin{equation}
f(t) = g(t) + O(h(t)).
\end{equation}
We note that if a function $f(t)$ satisfies
\begin{equation}
\frac{df}{dt}=O(e^{\lambda kt}), \qquad  {\text {then}} \qquad\left \{
\begin{array}{ll}
f=O(1) & \textmd{if } k<0 \\
f=O(t) & \textmd{if } k=0 \\
f=O(e^{\lambda kt}) & \textmd{if } k>0
\end{array}\right.
\end{equation}
\subsection{Spatial homogeneity and Wald's theorem}
We recall that an initial data set for the Einstein equations on a spacelike hypersurface ${\cal{S}}$ is called locally homogeneous if the naturally associated data set on the universal covering
 manifold $\tilde {\cal{S}}$ is homogeneous i.e. invariant under a transitive group action. The universal cover of the given space-time may not be spatially homogeneous but it can be extended to be so, in which case it has a preferred foliation by orbits with each leaf having constant mean curvature. The foliation is then topologically of the form ${\cal{S}}\times I$, where $I\subset\SYM \R$. 
Spatially homogeneous space-times then admit a $3-$dimensional group of isometries which acts on the
spacelike hypersurfaces ${\cal{S}}$. If the group acts simply (resp. multiply) transitively, the space-time is called Bianchi (resp. Kantowski-Sachs) space-time. In this paper, we assume that the group acts simply transitively.

Bianchi space-times can be classified according to the
Lie algebra of Killing vector fields and its associated isometry group $G_3$. 
The problem then reduces
to classifying the Lie algebra structure constants 
which satisfy the algebraic restrictions given by the Jacobi identities. Let $\bm \xi_\a,
\a=1,2,3$  
represent the basis of Killing vector fields associated to the group $G_3$.   
The Killing vector fields satisfy $[ {\bm \xi_\a}, {\bm\xi_\b}]= C^{\c}_{~\a\b} {\bm \xi_\c}$, where $C\tensor{}{\c}{\a\b} $ are the structure constants of the Lie algebra.
If $\bm n$ denotes the future-pointing unit vector normal to the surfaces of homogeneity, then ${\bf n}$ is group-invariant and $[\bm n, \bm\xi_\b]=0.$

In this article, we consider the orthonormal frame formalism described in Appendix \ref{sec1+3}. In particular, we consider a basis $\{\bm e_0,\bm e_\a\}$ where the timelike vector $\bm e_0$ is chosen to coincide with ${\bf n}$. The triad of spatial vectors $\bm e_\a$ can be chosen to be tangent to the group orbits and to commute with the Killing vector fields, 
$[\bm e_\a, \bm\xi_\b]=0.$
We will refer to an observer in a Bianchi space-time, who follows the integral curves of $\bm e_0$ as a \textit{canonical observer}.
Since the congruence $\bm e_0$ is hypersurface orthogonal, the vector fields $\bm e_\a$ generate a Lie algebra (with structure constants say $\gamma^\a_{~\b\g}$) which is equivalent to the Lie algebra of the Killing vector fields (with structure constants $C^\a_{~\b\g}$) \cite{MacCallum73}. One can thus classify the Bianchi space-times using the structure constants of either algebra, which we decompose as   
\begin{eqnarray}
\g^\a{}_{\b\g}=\e_{\b \g \delta} n^{\delta \a}+\delta^\a_{~\g} a_\b-\delta^\a_{~\b} a_\g,
\end{eqnarray}
where $n^{\a \b}=n^{(\a \b)}$ and $\varepsilon_{\b \g \delta}$ and $\delta^\a_{~\b}$ are the Levi--Civita and
the Kronecker delta symbols, respectively.

Now, one finds, using the 
Jacobi identities, that the rotation coefficients in the commutators are spatially constant i.e.
${\bf e}_\a(\gamma^a_{~bc})=0.$
Therefore, the Jacobi identities (\ref{ola22}) give
\begin{equation}
\label{n&a1}
n^{\a\b}a_\b=0.
\end{equation}
One still has the freedom, in the frame choice, of
a time-dependent rotation of the triad ${\bm e_\a}$ in a surface of constant coordinate time which
is usually used to set:
\begin{equation}
\label{n&a}
n_{\a\b}={\rm diag}(n_1,n_2,n_3),~~~a_\a=(a,0,0)
\end{equation}  
and one can divide the 
Lie algebras into class $A$ ($a=0$) and $B$ ($a\ne 0$), which correspond  to the unimodular and non-unimodular Lie algebras, respectively. However, we will not make this choice here.

By choosing our orthonormal frame ${\bf e}_a$ such that ${\bf n} = {\bf e}_0$, 
it follows that ${\bf n}$ is tangent to a congruence of curves with vanishing acceleration $A^a $ and vanishing vorticity $\omega_{ab}$. The structure constants $\Acf_a$ and $\Ncf_{ab}$ are equivalent to those derived from the basis of Killing vectors.
Thus, in our setting, a \emph{canonical observer} in the Bianchi space-time thus experiences no acceleration $A^\mu$, no vorticity $\omega^\alpha $ and uses a Fermi-propagated frame, i.e. the quantity $\Omega_\alpha$ defined in \eqref{Omega_definition} vanishes. We will refer to this as our {\em gauge choice}.

Now, it is known that ever expanding spatially homogeneous space-times are geodesically complete, in particular all inextendible causal geodesics are complete in the future direction \cite{Lee, Ren95}. Another important known result states that  the de Sitter solution is an attractor for ever-expanding Bianchi space-times \cite{Wal83}:
\\\\
{\bf Theorem (Wald):} 
\\
{\em  Consider a Bianchi space-time $(M,g)$, with a positive cosmological constant $\Lambda$, initially expanding and satisfying the dominant and strong energy conditions. If $(M,g)$ is of Bianchi types I-VIII, 
then $(M,g)$ locally asymptotically approaches de Sitter.}
%
\subsection{Preliminary estimates}
\label{sec2.2}
Following Wald \cite{Wal83}, throughout this article we assume that:
\begin{assumption}
	\label{Assumptions Wald+}
	~
	\begin{enumerate}
		\item[1.1] The cosmological constant $\Lambda$ is strictly positive and we set $\lambda = \sqrt{\third \Lambda}$.
		\item[1.2] The matter model satisfies the DEC and SEC.
		\item[1.3] The space-time is initially expanding.
		\item[1.4] The space-time is spatially homogeneous of Bianchi type and the scalar curvature of the surfaces of homogeneity satisfies $^{(3)}\hspace{-0.05cm}R \le 0$.
	\end{enumerate}
\end{assumption}
\begin{remark} We note that $^{(3)}\hspace{-0.05cm}R$ is given explicitly in \eqref{efe7SH} in terms of the structure coefficients.
	Condition  $^{(3)}\hspace{-0.05cm}R \le 0$ is automatically satisfied for all Bianchi types except type IX.
\end{remark} 
In order to study the conformal regularity and future stability, we must prove strong enough decay estimates for the background variables. Based on \cite{Wal83} and \cite{Ren95}, some useful asymptotic decay estimates were already given by Lee \cite{Lee} for the metric and shear tensors, under Assumptions \ref{Assumptions Wald+}. However, those estimates are not sharp enough for our purposes and, thus, we need to improve them first.

For our analysis we shall use the Einstein equations \eqref{Einstein equation} expressed in terms of kinematic quantities \eqref{ola73}, \eqref{ola10} and group structure quantities \eqref{n&a}. See Appendix \ref{appendix-equations} for details. 

In the case of spatially homogeneous background space-times, the Einstein equations are ODEs having time $t$ as independent variable.    
\subsubsection*{Expansion}
Throughout this article we use the expansion variable 
$$\Hubble:=\frac{1}{3}\theta=\frac{1}{3}\nabla_\mu(e_0^\mu).$$
From \eqref{efe0SH} and \eqref{efe6SH}, using Assumptions 1.1, 1.2 and 1.4, we derive
\begin{equation}
\label{Hubble inequalities}
{\bf e}_0(\Hubble) \le \lambda^2 - \Hubble^2 \le 0.
\end{equation}
By Assumption 1.3, we have that 
\begin{equation}
\label{Hubble simple bound}
\Hubble_*:=\Hubble(t=t_*) >0 \, \Rightarrow \, \Hubble_* \ge \Hubble(t) > \lambda >0, \quad\forall t\ge 0.
\end{equation}
Thus, an initially expanding universe is eternally expanding. 
Using  \eqref{Hubble inequalities} in two different ways, one finds
\begin{equation}
\label{Hubble lower and upper bound}
\lambda \le \Hubble \le \lambda \coth (\lambda t) \quad \implies \quad H =\lambda+O(e^{-2\lambda t}),
\end{equation}
which, in turn, implies
\begin{equation}
\label{bound for lambda squared minus Hubble squared}
\Hubble^2 - \lambda^2 = (\Hubble-\lambda)( \Hubble + \lambda) = O(e^{-2\lambda t}).
\end{equation}
\subsubsection*{Length scale}
We define the length scale $L$ by 
$$H=\frac{1}{L}\frac{dL}{dt}$$
which implies that $L$ is a strictly increasing function with
\begin{equation}
\label{definition of average length scale}
L=L_* e^{\int_0^t H(s) ds}.
\end{equation}
The lower bound of \eqref{Hubble lower and upper bound} gives $L_* e^{\lambda t} \le L  $. Using the upper bound of \eqref{Hubble simple bound} on an initial time interval, and the upper bound of \eqref{Hubble lower and upper bound} for a late time interval, implies that for some constant $C$
\begin{equation}
\label{bounds for L}
L_* e^{\lambda t} \le L \le C e^{\lambda t}, \quad \forall t>0.
\end{equation}
Hence, 
$L e^{-\lambda t} $ is bounded away from zero, which is a stronger statement than $L=O(e^{\lambda t})$.
\subsubsection*{Shear: intermediate decay rate}
Combining \eqref{efe6SH} with \eqref{Hubble lower and upper bound} implies
$$
0 \le \sigma^2 \le 3(H^2 - \lambda^2) \le \frac{3\lambda^2}{\sinh^2(\lambda t )}
$$
and
\begin{equation}
\label{bound for shear 1}
 \vert \sigma_{\alpha\beta} \vert \le \frac{\sqrt{6}\lambda}{\sinh(\lambda t )}
 \quad \implies \quad \sigma_{\alpha\beta}  = O(e^{-\lambda t}). 
\end{equation}
\subsubsection*{Matter contents}
By similar methods, we find an intermediate bound for the density
\begin{equation}
\label{bounds for density 1}
\rho = T_{\a\b} n^\a n^\b = T_{00} \le 3(H^2-\lambda^2)  \quad \implies \quad \rho=O(e^{-2\lambda t}).
\end{equation}
The DEC implies that $T_{00}$ dominates the other components of $T_{\a\b}$ (see e.g. \cite{HawkingEllis}, page 91) and, hence, we get the same intermediate bound for $q_\a, \pi_{\a\b}$ and $p$. Hence, for a spatially homogeneous space-time with a matter model satisfying Assumptions 1.2, we find that
\begin{equation}
\label{minimum matter bound}
\rho, \,\, p, \,\, q_\a, \,\,  \pi_{\a\b} =O(e^{-2\lambda t})\quad \implies \quad T_{\a\b}=O(e^{-2\lambda t}).
\end{equation}
\subsubsection*{Spatial connection coefficients and 3-curvature}

We define $\ACF_{\alpha} := L \Acf_{\alpha}  $ and $\NCF_{ab} := L \Ncf_{ab}  $ and rewrite 
\eqref{jac4SH} and \eqref{jac5SH}
as
\begin{eqnarray}
\label{Acf2 equation}
e_0(\ACF_{\alpha})  &=&  - \sigma\tensor{\alpha}{\beta}{} \ACF_{\beta}   \\
\label{Ncf2 equation}
e_0(\NCF_{\alpha\beta}) &=& 2 \sigma\tensor{(\alpha}{\gamma}{} \NCF_{\beta)\gamma},
\end{eqnarray}
where $\ACF^\alpha$ is a vector in a Euclidean 3-space with norm $\Vert \cdot \Vert$. Its length $z=\Vert \ACF \Vert$ satisfies
\begin{equation}
\label{z with shear norm}
z\dot{z} = - \sigma\tensor{}{\alpha\beta}{} \ACF_\alpha \ACF_{\beta} 
\le C \Vert  \sigma \Vert z^2,
\end{equation}
for some constant $C>0$, where $\Vert  \sigma \Vert$ represents the euclidean norm of a matrix with entries $\sigma_{\alpha\beta}$.
Since $\Vert  \sigma \Vert$ is bounded by \eqref{bound for shear 1}, then
$e^{\int_0^t  \Vert  \sigma \Vert  ds}=O(1) $. This implies that $z=O(1) $, from which it follows that $\ACF_\alpha=O(1)$. A similar argument can be used to derive $\NCF_{\alpha\beta}=O(1)$.

We, hence, have the following bounds for the spatial connection coefficients
\begin{eqnarray}
\label{bound Acf}
\Acf_\alpha &=& O(e^{-\lambda t}),\\
\label{bound Ncf}
\Ncf_{\alpha\beta} &=& O(e^{-\lambda t}).
\end{eqnarray}
\subsubsection*{Shear: improved rate}
Using the estimates obtained so far, \eqref{efe2SH} can be written in the form
\begin{equation}
\label{asdfg}
{\bf e}_0(\sigma^{\a\b})=-3\Hubble\sigma^{\a\b} + O(e^{- 2 \lambda t}),
\end{equation}
which implies ${\bf e}_0(L^3\sigma^{\a\b}) = O(e^{ \lambda t})$, and then
\begin{equation}
\sigma^{\a\b} = O(e^{- 2 \lambda t}).\label{estimates-shear}
\end{equation}
Hence, $\Sigma_{\a\b} :=L^2 \sigma_{\a\b} = O(1)$ is bounded.
\begin{remark}
	When the shear can be diagonalised as $diag(\sigma_1, \sigma_2, \sigma_3) $,
	the second fundamental form 
	$\theta\tensor{\a}{\b}{} = \Hubble \delta\tensor{\a}{\b}{} +\sigma\tensor{\a}{\b}{} $ is diagonal and has eigenvalues $ \lambda_\a = \Hubble + \sigma_\a$.
	Bounds \eqref{Hubble lower and upper bound} and \eqref{estimates-shear} then imply that the generalised Kasner exponents $p_\a = \lambda_\a/(3H)$ satisfy
	\begin{equation}
	\label{Kasner exponents}
	p_\a = \frac{1}{3} + O(e^{-2\lambda t}).
	\end{equation}
\end{remark}	
\section{Unphysical variables and improved estimates}\label{Section unphysical variables}

In the estimates \eqref{asdfg}-\eqref{estimates-shear}, we observed that the dominant term on the right hand side, the term including $\Hubble$, could be eliminated by rescaling the variable by a suitable power of $L$. As a result, we could exploit the decay rates of the remaining terms in the evolution equations and obtain new estimates by integration.

For that reason, we will introduce unphysical rescaled quantities. Our guiding principle will be to eliminate the terms including $\Hubble$ from the evolution equations. 
We note that some of these unphysical quantities play an important role in the conformal geometry of the space-time. We will discuss this aspect in more detail in Section \ref{Section - Conformal regularity} where we use it to prove certain notions of conformal regularity.
\subsubsection*{Unphysical variables}
Above, we defined the variables
\begin{equation}
\label{ACF_NCF_Sigma}
  \ACF_\a := L \Acf_\a, \qquad 
  \NCF_{\a\b} := L \Ncf_{\a\b}, \qquad 
  \Sigma_{\a\b} :=L^2 \sigma_{\a\b}
\end{equation}
and established that they are bounded, i.e.  
\begin{equation}
\label{ACF_NCF_Sigma bounded} 
\ACF_{\a} ,\NCF_{\a\b} ,\Sigma_{\a\b} = O(1) .
\end{equation}
We define the following additional rescaled variables:
\begin{align}
  ^{(3)}\hspace{-0.05cm}\hat S_{\a\b} &:= L^2 \,\,^{(3)}\hspace{-0.05cm}S_{\a\b} & ^{(3)}\hspace{-0.05cm}\hat R:= L^2 \,\,^{(3)}\hspace{-0.05cm}R \nonumber\\
  \label{unphysical variables defined}
  {\cal E}_{\a\b}&:=L^3E_{\a\b} &{\cal H}_{\a\b}:=L^3H_{\a\b}\\
  Q_\a&:=L^3q_\a &\Pi_{\a\b}:=L^3\pi_{\a\b}, \nonumber
\end{align}
where $E_{\a\b}$ and $H_{\a\b}$ are the electric and magnetic parts of the Weyl tensor, as defined in Appendix \ref{sec1+3}, and $^{(3)}\hspace{-0.05cm}S_{\a\b}$ is the trace free part of the 3-Ricci curvature as given in Appendix \ref{appendix-equations}. 
Our aim is to use the Einstein equations to derive estimates for all these quantities. 
\subsubsection*{Conformal frame and conformal time}
For our analysis, we define the rescaled frame $\hat {\bf  e}_a $ and the conformal time $\tau$ by
\begin{equation}
\label{conformal time}
\hat{\bf e}_a := L {\bf e}_a\,, \qquad\qquad\tau := \int_0^t \frac{1}{L(s)} ds.
\end{equation}
The bounds in \eqref{bounds for L} imply that $\tau$ is a bounded, strictly increasing function of $t$ with limit $\tau_\infty$ as $t\to \infty$ where
$$
\frac{1}{C \lambda} \le \tau_\infty \le \frac{1}{L_* \lambda}.
$$
Thus, the infinite time interval $[0,\infty)$ for $t$ is mapped to the finite interval $[0,\tau_\infty)$ for $\tau$. As a result, we are able to 
study the asymptotic behaviour in terms of the local behaviour near $\tau=\tau_\infty$.

The analysis of the conformal boundary in Section \ref{Section - Conformal regularity} will require functions to have a finite limit as $\tau\to\tau_\infty$. Hence, boundedness is insufficient for our purposes and we will focus part of our analysis on showing that a given quantity $F$ has a well-defined limit $F(\tau_\infty)$ as $\tau\to\tau_\infty$. We observe that for $k<0$, $F=O(L^k) \implies F(\tau_\infty)=0$, whereas ${\bf e}_0(F)=O(L^k)$ with $k<0$ only provides boundedness on $F$. However, for our purposes it is sufficient to show that $\hat{\bf e}_0 (F)=O(1)$, as this implies that $F(\tau_\infty)$ is finite and well defined. Hence, we will often work with derivatives along $\hat{\bf e}_0$ in order to deduce the existence of the relevant limits. Due to the rescaling of the variables \eqref{ACF_NCF_Sigma} - \eqref{unphysical variables defined} we have a regular occurrence of powers of $L$ in the evolution equations and, hence, we will characterize decay rates in terms of $L$ using $O(e^{k \lambda t}) \equiv O(L^k)$.
\subsection{The unphysical constraint and evolution equations}
We now consider the rescaled Einstein field equations. Substituting \eqref{ACF_NCF_Sigma}-\eqref{unphysical variables defined} into \eqref{efe70SH}-\eqref{efe1SH} and \eqref{ola26SH}-\eqref{ola28SH} leads to the following constraint equations
\begin{eqnarray}
\label{efe70SHR}
\hat S_{\a\b} &=& 2\hat n_{\a\gamma}\hat n^{\gamma}_{~\beta}-\hat n^{\gamma}_{~\gamma}\hat n_{\a\b}-\frac{1}{3}\delta_{\a\b} [ 2\hat n^{\a\b}\hat n_{\a\b}-(\hat n^\a_{~\a})^2 ] + 2 \varepsilon^{\gamma\delta}_{~~(\a}\hat a_{|\gamma|} \hat n_{\b)\delta}
\\
\label{efe7SHR}
\hat R &=& -6\hat a_\a \hat a^\a-\hat n^{\a\b}\hat n_{\a\b}+\frac{1}{2}(\hat n^\a_{~\a})^2
\\
\label{efe1SHR}
Q^\a &=& 3\Sigma^{\b\a} \hat a_\b - \varepsilon^{\a\b\gamma} \hat n_{\gamma\delta}\Sigma_{~\b}^\delta
\\
\label{ola26SHR}
{\cal E}_{\a\b}&=&-\half\Pi_{\a\b}
+L(\Hubble \Sigma_{\a\b}+\hat S_{\a\b})
-\frac{1}{L} \left(\Sigma_{\a\g}\Sigma^\g_{~\b} - \frac{2}{3} \delta_{\a\b}\Sigma^2\right) 
\\
\label{ola28SHR}
{\cal H}_{\a\b}&=&
\frac{1}{2} \hat n^\g_{~\g}\Sigma_{\a\b} - 3 \hat n_{~(\a}^{\g}\Sigma_{\b)\g} + \delta_{\a\b} \hat n_{\g\delta}\Sigma^{\g\delta}
- \varepsilon^{\g\delta}{}_{\a}\hat a_{|\g|} \Sigma_{\b)\delta},
\end{eqnarray}
where $\Sigma^2=\frac{1}{2}\Sigma_{\a\b}\Sigma^{\a\b}$.

Substitution of \eqref{ACF_NCF_Sigma bounded} into \eqref{jac4SH}, \eqref{jac5SH}, \eqref{efe51SH}, \eqref{efe2SHrewritten}, \eqref{EpropSH} and \eqref{HpropSH} leads to the following propagation equations with respect to $\hat{\bf e}_0$:
\begin{eqnarray}
\label{jac4SHR}
{\bf \hat e}_0 (\ACF^\a)
&=&
-\frac{1}{L} \ACF_\b \Sigma^{\a\b}
\\
\label{jac5SHR}
{\bf \hat e}_0 (\NCF^{\a\b}) 
&=&
-\frac{1}{L} \NCF^{(\a}_{~~~\g} \Sigma^{\b)\g}
\\
\label{efe2SHR}
{\bf \hat e}_0(\Sigma^{\a\b})
\label{efe2SHR2}
&=&
- {\cal E}_{\a\b} +\frac{1}{2} \Pi^{\a\b}
-\frac{1}{L}\Sigma_{\a\g}\Sigma^\g_\b +\frac{2}{3L} \delta_{\a\b}\Sigma^2 
\\
\label{efe51SHR2}
{\bf \hat e}_0 (LQ_\a)&=&
3L\hat a_\b\Pi^\b_{~\a} + L\Pi^\gamma_\b\varepsilon_{\gamma\a\delta}\hat n^{\b\delta} - \Sigma_{\a\b}Q^\b 
\\
\label{EpropSHR}
{\bf \hat e}_0({\cal E}^{\alpha\beta}+\frac{1}{2}\Pi^{\alpha\beta})
&=&
\Hubble L\Pi^{\a\b} + \frac{3}{L}\Sigma^{(\alpha}_{~~\gamma}({\cal E}^{\beta)\gamma}-\frac{1}{6}\Pi^{\beta)\gamma})-\frac{1}{L}\delta^{\alpha\beta} \Sigma_{\gamma\delta}({\cal E}^{\gamma\delta}-\frac{1}{6}\Pi^{\gamma\delta})\nonumber\\
&&-\varepsilon^{\gamma\delta(\alpha}\hat a_\gamma {\cal H}^{\beta)}_{~\delta}+
\frac{1}{2}\hat n^{\gamma}_{~\gamma}{\cal H}^{\alpha\beta}-3\hat n^{(\alpha}_{~~\gamma}{\cal H}^{\beta)\gamma}+\delta^{\alpha\beta}\hat n_{\gamma\delta}{\cal H}^{\gamma\delta}\nonumber\\
& &-\frac{1}{2}L^2( \rho + p)\Sigma^{\a\b} -\frac{1}{2} \hat a^{(\alpha}Q^{\beta)}+\frac{1}{6}\delta^{\alpha\beta} \hat a_\gamma Q^\gamma  +\varepsilon^{\gamma\delta(\alpha}\left(\frac{1}{2} \hat n^{\beta)}_{~~\gamma}Q_\delta \right)
\\
\label{HpropSHR}
{\bf \hat e}_0({\cal H}^{\a\b}) 
&=&
\frac{3}{L}\Sigma^{(\alpha}_{~~\gamma}{\cal H}^{\beta)\gamma}
-\frac{1}{L}\delta^{\a\b}\Sigma_{\gamma\delta}{\cal H}^{\gamma\delta}
+ \frac{1}{L}\varepsilon^{\gamma\delta(\alpha} \Sigma^{\beta)}_{~~\gamma} Q_\delta 
+\varepsilon^{\gamma\delta(\alpha} \hat a_\gamma({\cal E}^{\beta)}_{~~\delta}-\frac{1}{2}\Pi^{\beta)}_{~~\delta})
\nonumber
\\
&&\hspace{-0.45cm}-\frac{1}{2}\hat n^\gamma_{~\gamma}({\cal E}^{\a\b}-\frac{1}{2}\Pi^{\a\b})
-\delta^{\a\b}\hat n_{\gamma\delta}({\cal E}^{\gamma\delta}-\frac{1}{2}\Pi^{\gamma\delta}) 
+3 \hat n^{(\alpha}_{~~\gamma}({\cal E}^{\beta)\gamma}-\frac{1}{2}\Pi^{\beta)\gamma})\label{fourty}
\end{eqnarray}
\subsection{Decay rates for Assumption 1}

Our aim is to improve the decay estimates of Section 2.2. 
We start our analysis using only Assumption \ref{Assumptions Wald+} before making additional assumptions to derive stronger decay rates.

Using the estimates found so far, we get from \eqref{jac4SHR} and \eqref{jac5SHR}
\begin{eqnarray}
{\bf \hat e}_0 (\hat a^\a)
= O(L^{-1}),
\quad\quad
{\bf \hat e}_0 (\hat n^{\a\b})
= O(L^{-1}) \nonumber .
\end{eqnarray}
Integrating these equations implies that $\ACF_\a$ and $\NCF_{\a\b}$ are finite at $\tau=\tau_\infty$ as are $\hat{S}_{\a\b} $ and $\hat R$ by \eqref{efe70SHR} and \eqref{efe7SHR}. Hence, the space-time metric induces a finite intrinsic connection and 3-curvature on the hypersurface $\tau = \tau_\infty $.

The unphysical constraints \eqref{efe1SHR}-\eqref{ola28SHR} combined with $ \Pi_{\a\b} = O(L)$, imply that 
\begin{equation}
\label{basic decay rates Q E H}
Q_\a = O(1), \quad {\cal E}_{\a\b} = O(L), \quad {\cal H}_{\a\b} =O(1).
\end{equation}
Combining these decay rates with \eqref{Hubble lower and upper bound}, \eqref{minimum matter bound} and \eqref{estimates-shear}, and using the definitions of the unphysical variables \eqref{ACF_NCF_Sigma}-\eqref{unphysical variables defined}, we readily obtain the following decay rates: 
\begin{proposition}
	\label{prop-decay1}
	Suppose Assumption \ref{Assumptions Wald+} holds, then as $t\to\infty$	
	\begin{enumerate}
		\item $\hat a_\b,\, \hat n_{\a\b},\, \hat S_{\a\b},\, \hat R$  all have finite limits.
		\item $\Sigma_{\a\b},\, Q_\a,\, {\cal H}_{\a\b}$ are bounded.
		\item ${\cal E}_{\a\b}=O(L)$.
	\end{enumerate}
	Hence, the following decay rates hold for the physical variables:
	
	\medskip
	\begin{tabular}{lclclcl}
		$\Acf_\a = O(e^{-\lambda t})$ &&
		$\Ncf_{\a\b} = O(e^{-\lambda t})$ &&
		$^{(3)}\hspace{-0.05cm}S_{\a\b} = O(e^{-2\lambda t})$ &&
		$^{(3)}\hspace{-0.05cm}R = O(e^{-2\lambda t})$ \\
		&& && &&\\
		$\Hubble = \lambda + O(e^{-2\lambda t})$ &&
		$\sigma_{\a\b} = O(e^{-2\lambda t})$ &&
		$E_{\a\b} = O(e^{-2\lambda t})$ &&
		$H_{\a\b} = O(e^{-3\lambda t})$ \\
		&& && &&\\
		$\rho = O(e^{-2\lambda t})$ &&
		$p = O(e^{-2\lambda t})$ &&
		$q_\a = O(e^{-3\lambda t})$ &&
		$\pi_{\a\b} = O(e^{-2\lambda t})$.
	\end{tabular}
\end{proposition}
\begin{remark}
This result improves some of the decay rates obtained in \cite{Wal83, Lee} and \cite{Ren95}. For example, the rate for the shear in \cite{Lee} is $O(e^{-\lambda t})$. Our improvement is crucial to sharpen the decay rates for the Weyl tensor and, in turn, to apply the conformal methods in sections 5 and 6.
\end{remark}
The estimates established above imply that, in the orthonormal frame of an observer, the components of the space-time curvature and the curvature of the surfaces of homogeneity decay like
\begin{equation}
\label{CNHT curvature decays}
C_{abcd} = O(e^{-2\lambda t}), \quad 
R_{ab} = \Lambda g_{ab} + O(e^{-2\lambda t}), \quad
^{(3)}\hspace{-0.05cm}S_{ab} = O(e^{-2\lambda t}), \quad
^{(3)}\hspace{-0.05cm}R = O(e^{-2\lambda t}).
\end{equation}
At late times, the observer will be less and less able to distinguish the curvatures of the Bianchi space-time from those of the de Sitter space-time. Hence, the above provides a precise form in which the Bianchi space-times considered approach, locally in space, de Sitter space-time at late times \cite{Wal83}.
The fact that the space-time approaches de Sitter asymptotically and, in particular, that the structure constants $\Acf_\a, \Ncf_{\a\b}$ decay to zero could suggest that any information about the space-time Bianchi type is ultimately lost as the Lie algebra of the surfaces of homogeneity approach that of Bianchi type I. However, as we  will show in Section \ref{sec:conformal hair}, the rescaled structure constants $\ACF_\a, \NCF_{\a\b}$ preserve the information about the Bianchi type up to and beyond $\tau_\infty$.

Given the estimates of Proposition \ref{prop-decay1}, the remaining evolution equations \eqref{efe2SHR2}-\eqref{fourty} take the form 
\begin{eqnarray}
\label{efe2SHREst}
{\bf \hat e}_0(\Sigma^{\a\b})&=&
- {\cal E}_{\a\b} +\frac{1}{2} \Pi^{\a\b} + O(L^{-1})
\\
\label{efe51SHR2Est2}
{\bf \hat e}_0 (LQ_\a)&=&
3L\hat a_\b\Pi^\b_{~\a} + L\Pi^\gamma_{~\b}\varepsilon_{\gamma\a\delta}\hat n^{\b\delta} + O(1)
\\
\label{EpropSHR2Est}
{\bf \hat e}_0({\cal E}^{\alpha\beta}+\frac{1}{2}\Pi^{\alpha\beta})
&=&  
\Hubble L\Pi^{\a\b}  + O(1)
\\
\label{HpropSHR2Est}
{\bf \hat e}_0({\cal H}^{\a\b}) 
&=&
3 \hat n^{(\alpha}_{~~\gamma}({\cal E}^{\beta)\gamma}-\frac{1}{2}\Pi^{\beta)\gamma})-\frac{1}{2}\hat n^\gamma_{~\gamma}({\cal E}^{\a\b}-\frac{1}{2}\Pi^{\a\b})-\delta^{\a\b}\hat n_{\gamma\delta}({\cal E}^{\gamma\delta}-\frac{1}{2}\Pi^{\gamma\delta})\nonumber\\
&&+\varepsilon^{\gamma\delta(\alpha} \hat a_\gamma({\cal E}^{\beta)}_{~~\delta}-\frac{1}{2}\Pi^{\beta)}_{~~\delta}) +O(L^{-1}).
\end{eqnarray}
We note that, in general, we have no isolated evolution equation for the anisotropic stress.\footnote{Equation \eqref{EpropSHR} gives an evolution equation for ${\cal{E}}_{\a\b}+\half \Pi_{\a\b}$. However, in our view, the evolution of the Weyl curvature is determined by the matter content here.}
However, we can see that the decay of $\Pi_{\a\b}$ plays an important role in improving our decay rates of other variables, in particular for ${\cal{E}}_{\a\b}$, or could be an obstruction. 
In turn, $\Pi_{\a\b}$ depends on the matter fields to be considered and we will study particular cases in the next section. But, before that, and still keeping enough generality on the matter content, we will make additional assumptions regarding the decay of $\Pi_{\a\b}$ in order to derive stronger decay rates.
\subsection{Decay rates:  Case with vanishing anisotropic stress}
For certain matter fields, such as aligned perfect fluids and scalar field space-times, the anisotropic stresses vanish and the system of evolution equations simplifies significantly. We thus start our analysis with the case $\Pi_{\a\b}=0$ in order to gain better insight into the analysis required to derive improved decay rates in more general cases.

When $\Pi_{\a\b}=0$, the evolution equations \eqref{efe2SHR}-\eqref{EpropSHR} take the form:
\begin{eqnarray}
\label{efe2SHRnopi}
{\bf \hat e}_0(\Sigma_{\a\b})&=&
- {\cal E}_{\a\b}
-\frac{1}{L}\Sigma_{\a\g}\Sigma^\g_\b +\frac{2}{3L} \delta_{\a\b}\Sigma^2  
\\
\label{efe51SHRnopi}
{\bf \hat e}_0 (LQ_\a)&=&
- \Sigma_{\a\b}Q^\b 
\\
\label{EpropSHRnopi}
{\bf \hat e}_0({\cal E}^{\alpha\beta})
&=&
 \frac{3}{L}\Sigma^{(\alpha}_{~~\gamma}{\cal E}^{\beta)\gamma}-\frac{\delta^{\alpha\beta}}{L} \Sigma_{\gamma\delta}{\cal E}^{\gamma\delta}-\varepsilon^{\gamma\delta(\alpha}\hat a_\gamma {\cal H}^{\beta)}_{~\delta}+
\frac{1}{2}\hat n^{\gamma}_{~\gamma}{\cal H}^{\alpha\beta}-3\hat n^{(\alpha}_{~~\gamma}{\cal H}^{\beta)\gamma}\\
&&+\delta^{\alpha\beta}\hat n_{\gamma\delta}{\cal H}^{\gamma\delta}
-\frac{L^2}{2}( \rho + p)\Sigma^{\a\b} -\frac{1}{2} \hat a^{(\alpha}Q^{\beta)}+\frac{1}{6}\delta^{\alpha\beta} \hat a_\gamma Q^\gamma  +\varepsilon^{\gamma\delta(\alpha}\left(\frac{1}{2} \hat n^{\beta)}_{~~\gamma}Q_\delta \right)\nonumber
\end{eqnarray}
Using the results of Proposition \ref{prop-decay1}, we see that ${\bf \hat e}_0 (LQ_\a)$ is bounded. Thus
$LQ_\a$ is finite at $\tau_\infty$ and hence $Q_\a $ must vanish at $\tau_\infty$. 
Similarly,  ${\bf \hat e}_0({\cal E}^{\alpha\beta})$ is bounded and hence ${\cal E}^{\alpha\beta}$ is finite at $\tau_\infty$. These estimates, in turn, imply that ${\bf \hat e}_0(\Sigma_{\a\b})$ is finite at $\tau_\infty$ and so is $\Sigma_{\a\b}$. It follows directly from \eqref{ola28SHR} and \eqref{HpropSHR2Est}  that ${\cal H}^{\alpha\beta}$ and ${\bf \hat e}_0({\cal H}^{\alpha\beta})$ are finite at $\tau_\infty$.
Note that for ${\bf \hat e}_0 ({\cal E}^{\alpha\beta})$ to be finite we require $L^2(\rho+p)$ to have a finite limit. As we shall see, this will be the case of e.g. perfect fluids with a linear equation of state. 
\subsection{Decay rates: Case with slowly decaying anisotropic stress} 
\label{sec3.4}

For some matter models that we will consider ahead, such as the Einstein-Maxwell fields, the anisotropic stress does not vanish.
Since we assume that the density $\rho$ dominates over the anisotropic stress, then $\pi_{\a\b}$ always decays at least as fast as $\rho$. For some specific matter models, we may be able to use the specific form of the energy-momentum tensor to improve this decay rate even further. 
Here, we consider the following general assumption on the decay rate for $\Pi_{\a\b}$:
\begin{assumption}
	\label{Assumption for Pi}
	The rescaled anisotropic stress $\Pi_{\a\b}$ satisfies $\Pi_{\a\b}=O(L^k)$ for some value $k<1$.
\end{assumption}
This assumption is equivalent to $\pi_{\a\b}=O(e^{m \lambda t})$ with $m<-2$, as $m=k-3$. Deriving or justifying such a decay rate will depend on the individual matter model. Throughout the rest of this article we will see that the strength of the results depend specifically on the value of $k$ and, hence, we often proceed on a case by case basis.
\begin{proposition}
\label{prop-decay2}
	Suppose Assumptions 1 and 2 hold. Then, the decay rates of Proposition \ref{prop-decay1} hold and $\Sigma_{\a\b}, {\cal{H}}_{\a\b}, Q_\a $ have finite limits at $\tau_\infty$. Furthermore, depending on the value of $k$ the following decay rates hold:
	\begin{enumerate}
		\item If $0<k<1$, then
			${\cal{E}}_{\a\b} = O(L^k)$.
		\item If $k=0$, then ${\cal{E}}_{\a\b} = O(t)$. 
		\item If $-1<k<0$, then 
		\begin{enumerate}
			\item ${\cal{E}}_{\a\b} = O(1)$.
			\item $Q_\a = O(L^{k})$ and hence $Q_\a=0$ at $\tau_\infty$.
		\end{enumerate} 
		\item If $k\le-1$ or $\pi_{\a\b}=0 $, then
		\begin{enumerate}
			\item ${\cal{E}}_{\a\b}$ is finite at $\tau_\infty$.
			\item $LQ_\a $ is finite at $\tau_\infty$.
		\end{enumerate}
	\end{enumerate}      	
\end{proposition}
~
\proof 
~ 
\begin{enumerate}
	\item Equation \eqref{EpropSHR2Est} can be written as 
	${\bf e}_0({\cal{E}}_{\a\b} + \half \Pi_{\a\b}) = -\Hubble \Pi_{\a\b} + O(L^{-1})= O(L^k)$, for $0<k<1$.
	Integrating the previous equation we get
	 ${\cal{E}}_{\a\b}= O(L^k)$. Substituting this into \eqref{ola26SHR} implies
	$\lambda\Sigma_{\a\b} + \hat S_{\a\b} = O(L^{k-1})$.
	Thus, the left hand side vanishes at $\tau_\infty $. Since $\hat S_{\a\b} $ has a finite limit at $\tau_\infty $ so does $\Sigma_{\a\b} $. Now, it follows directly from \eqref{efe1SHR}, for $Q_\a$, and from \eqref{ola28SHR}, for ${\cal{H}}_{\a\b}$, that both quantities have finite limits.
	\item If $k=0$, then from \eqref{EpropSHR}, we get ${\bf e}_0({\cal{E}}_{\a\b} + \half \Pi_{\a\b})  = O(1)$ which implies $({\cal{E}}_{\a\b} + \half \Pi_{\a\b}) = O(t)$, respectively ${\cal{E}}_{\a\b} = O(t)$. Adapting the argument of the proof of point 1, we get finite limits for $\Sigma_{\a\b}, {\cal{H}}_{\a\b}$ and $Q_\a $.   
	
	\item If $k<0$, then the previous argument gives ${\cal{E}}_{\a\b} = O(1)$. Moreover ${\bf e}_0(LQ_\a) = O(L^{k})$ implies $Q_\a = O(L^{k-1}) $ and, hence, $Q_\a=0$ at $\tau_\infty$.
	
	\item If $k\le -1$, then \eqref{EpropSHR2Est} implies that $\hat{\bf e}_0({\cal{E}}_{\a\b} +  \frac{1}{2}\Pi_{\alpha\beta} )$ is bounded and hence ${\cal{E}}_{\a\b}+\frac{1}{2}\Pi_{\alpha\beta}$ is finite. Since $\Pi_{\a\b}\to 0$ by assumption, it follows that ${\cal{E}}_{\a\b}$ has a finite limit at $\tau_\infty$. Similarly, \eqref{efe51SHR2Est2} implies that $\hat{\bf e}_0(LQ_\a)$ is bounded and, hence, $LQ_\a$ is finite at $\tau_\infty$.
	
	We have already analysed the case $\pi_{\a\b}=0$ in detail. Alternatively, we note that it satisfies $\Pi_{\a\b}=O(L^k)$ for any $k<-1$.
	\endproof
\end{enumerate} 
The case $k\le -1$ allows us to establish the strongest decay rates and, in particular, that ${\cal{E}}_{\a\b}$ and ${\cal{H}}_{\a\b}$ are finite at $\tau_\infty$. Hence, in our analysis, we will focus on identifying matter models for which $\pi_{\a\b}=O(L^{-4})$.
\section{Decay rates for different matter models}
\label{Section - matter models}
We now discuss how different matter models satisfy the conditions of Proposition \ref{prop-decay2}. The matter models under consideration are Vlasov matter, scalar fields, perfect fluids, anisotropic matter (including viscous fluids and elastic matter) and trace-free matter models (including radiation fluids and Maxwell fields).
%
\subsection{Fluid space-times}
\label{Fluid space-times}
We suppose that Assumption \ref{Assumptions Wald+} holds and that we have a matter model for which 
$$\rho + p = \gamma \rho + O(\rho^2),$$
for some constant value of $\gamma$. We note that $q_\a$ and $\pi_{\a\b}$ need not vanish and an example will be considered in Section \ref{viscous-subsec}.
Our strategy is to deduce bounds on the density $\rho$ through repeated improvement of the decay rates, by feeding intermediate decay rates back into the evolution equation. 

Suppose $\rho=O(L^{-m})$ for some $m$, then $p,  q_\a, \pi_{\a\b}=O(L^{-m}) $. In particular, \eqref{minimum matter bound}  gives $m\ge 2 $ and $\rho^2=O(L^{-(m+1)})$ holds. 
In turn, the evolution equation \eqref{ids1SH} gives
\begin{eqnarray}
{\bf e}_0({\rho}) &=&- 3\Hubble \gamma\rho  + O(\rho^2)  -\pi_{\a\b}\sigma^{\a\b} + 2 a^\a q_\a   \nonumber \\
&=& - 3\Hubble \gamma\rho + O(L^{-(m+1)}), \nonumber 
\end{eqnarray}
implying
${\bf e}_0({\rho L^{3\gamma}}) = O(L^{3\gamma -(m+1)})$,
and therefore
\begin{equation}
\rho = \left\lbrace
\begin{array}{l c l}
O(L^{-3\gamma} ) &\quad & 3\gamma < (m+1) \\
O(L^{-3\gamma+\varepsilon} ) &\quad & 3\gamma = (m+1) \, \quad \varepsilon<<1\\
O(L^{-(m+1)}) &\quad & 3\gamma > (m+1). 
\end{array}\right. 
\end{equation}
For $3\gamma \ge m+1 $ we are able to improve the decay rate of $\rho$ by up to a factor of $L^{-1}$. We can then repeat the above steps with $\rho=O(L^{-(m+1)})$, respectively $\rho=O(L^{-(m+1-\varepsilon)})$, to obtain $\rho = O(e^{-3\gamma\lambda t})$.

Since Assumption 1.2 implies that $\rho$ dominates the other components of $T_{\a\b}$ we get 
\begin{equation}
\label{bounds for density perfect fluid}
\rho,\, p,\, q_\a,\, \pi_{\a\b} = O(e^{-3\gamma\lambda t}), \qquad \textmd{for}\quad  \frac{2}{3} \le \gamma \le 2.
\end{equation}
In fact,  we can show
\begin{equation}
\label{imperfect}
{\bf \hat e}_0(L^{3\gamma}\rho)= O(1),\quad {\text{hence}}\quad  L^{3\gamma}\rho \quad\textmd{is finite at} \,\,\tau_\infty.
\end{equation}
For $\gamma \ge \frac{4}{3}$ we can set $k=3-3\gamma \le -1$ in Proposition \ref{prop-decay2}.  Hence, for $\gamma \ge \frac{4}{3}$ the rescaled Weyl curvatures ${\cal E}_{\a\b}$ and ${\cal H}_{\a\b}$ are finite at $\tau_\infty$.
%
\subsection{Perfect fluids}
%
For a perfect fluid flowing along the timelike direction $\bf v$, the energy-momentum tensor has the form
\begin{equation}
\label{EM-tensor perfect fluid}
T^{\mathrm{fluid}}_{\mu\nu}=\bar{\rho} v_\mu v_\nu+\bar{p}(g_{\mu\nu}+v_\mu v_\nu),~~~v_\mu v^\mu=-1,
\end{equation}
where $\bar{\rho}, \bar{p}$ are the density and pressure measured by an observer co-moving with the fluid.

When the fluid flow $\bf v$ and the time direction ${\bf e}_0$ of the Bianchi space-time are parallel, the perfect fluid is referred to as {\em non-tilted} or {\em aligned}. Otherwise, we refer to a {\em tilted perfect fluid} (relative to ${\bf e}_0$) and, in that case, it has non-vanishing $q_\a$ and $\pi_{\a\b}$ as in \eqref{ola10}. In detail, we have
\begin{eqnarray}
\label{tilted density}
\rho&=& \Gamma^2(\bar \rho+\bar p)-\bar p
\\  
\label{tilted pressure} 
p&=& \frac{1}{3}(\bar \rho+\bar p)\Gamma^2\bar v^2+\bar p
\\  
\label{tilted energy flux}
q_\a&=& \Gamma^2(\bar \rho+\bar p)\bar v
\\  
\label{tilted fluid}
\pi_{\a\b} &=& \Gamma^2(\bar \rho+\bar p)(\bar v_a\bar v_b-\frac{1}{3}\bar v^2h_{ab}),
\end{eqnarray}
where $v^a=\Gamma(u^a+\bar v^a)$, $u_a\bar v^a=0$, $\Gamma:=1/\sqrt{1-\bar v^2}$ and $\bar v^2:=\bar v_a\bar v^a$. 
It thus follows that, as long as $v^a$ does not tilt towards a null vector at $\tau_\infty$, both $v^a$ and $\bar{\rho}$ are finite at $\tau_\infty$. 

Spatial homogeneity implies that a perfect fluid satisfies a barotropic equation of state $p=F(\rho)$ for a canonical observer in a Bianchi space-time (see e.g. \cite{EMMCbook}). 

For aligned perfect fluid space-times, the DEC and SEC are satisfied if $\rho \ge 0$, $\rho+p\ge 0 $ and $\rho+ 3p \ge 0 $. Hence, for aligned perfect fluids with a linear equation of state $p=(\gamma -1) \rho$, the DEC and SEC hold for $\gamma \in [\frac{2}{3},2] $.
For such perfect fluids, the evolution equation \eqref{ids1SH} gives directly
\begin{equation}
{\bf e}_0({\rho}) = - 3\Hubble \gamma\rho \quad 
\implies \quad \rho L^{3\gamma} = \textmd{const} \quad 
\implies \quad \rho = O(e^{-3\gamma \lambda t}).
\end{equation}
Hence, for aligned perfect fluids, the full energy momentum tensor decays like $e^{-3\gamma\lambda t}$. Since $\pi_{\a\b}=0$, we fall into case 4 of Proposition \ref{prop-decay2}.

Now, for an aligned perfect fluid with a non-linear equation of state of the form $\rho + p = \gamma \rho + O(\rho^2)$, we get the same decay rates as for the corresponding linear equation of state, as long as the DEC and SEC are satisfied. We note that a similar conclusion was reached in the appendix of  \cite{Rendall-1996} for some Bianchi I space-times with $p=f(\rho)$ such that $f(0)=0$, $0\le f'(\rho)\le 1$.
\subsection{Anisotropic matter}
\label{anisotropic matter}

Perfect fluids can be naturally generalised to anisotropic matter. Besides Maxwell fields and Vlasov matter, which we treat separately, another kind of physically interesting anisotropic matter are given by elastic matter models and viscous fluids. 

\subsubsection{Elastic matter} 

Examples of elastic matter in spatially homogeneous cosmologies were given in  \cite{Calogero-Heinzle}. In those models, the energy-momentum tensor is diagonal and has the form
\begin{equation}
\label{elastic}
T_{\mu\nu}=T^{\mathrm{PF}}_{\mu\nu} + \pi_{\mu\nu},
\end{equation}
where $T^{\mathrm{PF}}_{\mu\nu}$ has a perfect fluid form with $p=(\gamma -1) \rho=w\rho$ and $\pi_{\mu\nu}$ corresponds to the anisotropic elastic stress. Denoting the non-zero eigenvalues of $T\tensor{}{\mu}{\nu}$ by  $p_1, p_2, p_3$, we have $3p=p_1 + p_2 + p_3$. One can define
\begin{equation}
w_i = \frac{p_i}{\rho}
\end{equation}
so that $3w=w_1 + w_2 + w_3$. The DEC is satisfied if $w_i \in [-1,1]$, while the SEC holds if $w\ge -\frac{1}{3}$, i.e. $\gamma\ge \frac{2}{3}$.  Following the arguments of Section \ref{Fluid space-times}, we see that for $w \ge \frac{1}{3}$, i.e. $\gamma \ge \frac{4}{3}$ we achieve $k\le -1$. So from Proposition \ref{prop-decay2} the rescaled Weyl curvatures ${\cal{E}}_{\a\b}$ and ${\cal{H}}_{\a\b}$ have a finite limit at $\tau_\infty$. 

\subsubsection{Viscous fluids}
\label{viscous-subsec}
A relativistic generalization of the classical energy-momentum tensor for viscous fluids is given by \cite{Choquet}:
\begin{equation}
\label{viscous-T}
T_{\mu\nu}=T^{\mathrm{PF}}_{\mu\nu} + \pi_{\mu\nu}
\end{equation}
where $T^{\mathrm{PF}}_{\mu\nu}$ has a perfect fluid form and
\begin{equation}
 \pi_{\mu\nu}= \c_1 h_{\mu\nu} \nabla_\rho u^\rho + \c_2 h^\rho_{~\mu} h^\sigma_{~\nu}(\nabla_\rho u_\sigma+\nabla_\sigma u_\rho),
\end{equation}
where $\c_1$ and $\c_2$ are viscosity coefficients depending on the fluid under consideration.
A choice of $ \pi_{\mu\nu}$ which gives a Leray-Ohya hyperbolic system \cite{Choquet}, and is thus compatible with causal theory, is 
$ \pi_{\mu\nu}= \c_1 \theta g_{\mu\nu}+\c_2 \sigma_{\mu\nu}$.
Among this class of viscous fluids it is possible to find subclasses which are spatially homogeneous, irrotational and compatible with an orthornormal frame such that $\Omega_\mu=0$ and $n_{\mu\nu}$ is diagonal. These cases are considered in \cite{Ganguly} with
\begin{equation}
\label{viscous-example}
\pi_{\mu\nu}=\c_2 \sigma_{\mu\nu},
\end{equation}
and do not necessarily have a linear equation of state relating $\rho$ and $p$. In \cite{Ganguly}, the form considered was $p=(\gamma(\rho)-1)\rho$, with $\gamma(\rho)$ being a quartic function of $\rho$. Still, these equations of state fall in the class considered in Section \ref{Fluid space-times} which, for $\gamma \ge \frac{4}{3}$, satisfy the SEC and DEC and result in finite rescaled Weyl curvatures ${\cal{E}}_{\a\b}$ and ${\cal{H}}_{\a\b}$ at $\tau_\infty$.

\subsection{Trace-free matter}
\label{tracefree}
Suppose we have a matter model that satisfies the DEC and SEC and has a trace-free 
energy-momentum tensor\footnote{If the energy-momentum tensor is tracefree then the DEC and SEC are equivalent. However, tracefreeness does not imply that the DEC or SEC are satisfied, as the example of the conformal scalar field shows \cite{BarVis}.}.
 The latter implies that  $(\rho+p)=\frac{4}{3}\rho $, i.e. $\gamma=\frac{4}{3}$ in \eqref{bounds for density perfect fluid} and \eqref{imperfect}.
Hence, $\rho, p, q_\a, \pi_{\a\b}=O(L^{-4})$. More precisely, the quantities $\hat \rho = L^4 \rho, \, \hat p = L^4 p$ and $\hat q_\a = L^4 q_\a $ are all finite at $\tau_\infty$. Then, all trace-free matter satisfies Proposition \ref{prop-decay2} with $k=-1$ so that the rescaled Weyl curvatures ${\cal{E}}_{\a\b}$ and ${\cal{H}}_{\a\b}$ always have a finite limit at $\tau_\infty$.
%
\subsubsection{Radiation fluids (incoherent radiation)}
%
Radiation fluids are perfect fluids with $p=\tfrac{1}{3} \rho$ (i.e. $\gamma=\tfrac{4}{3}$). This equation of state holds for any observer since the energy-momentum tensor of a radiation fluid is trace-free, so the general estimates established above for trace-free matter hold for any (possibly tilted) radiation fluid. Hence, we can readily analyse radiation fluids that are tilted relative to the time-direction $\bf e_0$ using \eqref{tilted density}-\eqref{tilted fluid}.
\subsubsection{Null dust (pure radiation field)}
%
Null dust describes massless radiation and its energy-momentum tensor has the form
\begin{equation}
\label{null fluid}
T_{\mu\nu}=\rho k_\mu k_\nu,~~~\rho \ge 0, ~~ k_\mu k^\mu=0.
\end{equation}
Since null dust satisfies the DEC and SEC, the estimates above for trace-free matter hold.
%
\subsubsection{Einstein-Maxwell fields}
\label{EMYM section}
%
A space-time with a source-free electromagnetic field is described by
\begin{equation}
\label{ola11}
T_{\mu\nu}=F_{\mu\rho} F\tensor{\nu}{\rho}{}-\frac{1}{4}F^{\rho\sigma}F_{\rho\sigma}g_{\mu\nu},
\end{equation}
with the Faraday tensor $F_{\mu\nu}$ satisfying the source-free Maxwell equations
$$
\nabla^\mu F_{\mu\nu}=0,~~~\nabla_{[\rho}F_{\mu\nu]}=0.
$$
The Faraday tensor relative to ${\bf u}$ can be written in the form
\begin{equation}
\label{Faraday decomposition}
F_{ab} = u_a E_b - E_a u_b + \eta_{abcd} H^c u^d,
\end{equation}
 where the electric field $E_{a}=F_{ab}u^b$ and the magnetic field $H_a=\frac{1}{2}\eta_{abcd}F^{bc}u^d$ satisfy $E_a u^a=0$ and $H_a u^a=0$, and $\eta_{abcd}$ is the usual the 4-dimensional volume element.
Interpreting the source-free Maxwell field as a fluid flowing along ${\bf u}$ we have \cite{vElUgg96}
\begin{eqnarray}
\label{Maxwell density}
\rho &=& \half (E_\a E^\a + H_\a H^\a) = 3p\\
\label{Maxwell vector}
q^\a &=& \varepsilon^{\a\b\gamma} E_\b H_\gamma \\
\label{Maxwell TF tensor}
\pi_{\a\b} &=& - E_\a E_\b - H_\a H_\b + \third\delta_{\a\b}  (E_\gamma E^\gamma + H_\gamma H^\gamma) .
\end{eqnarray}
The energy-momentum tensor \eqref{ola11} is trace-free so that $\rho=O(L^{-4})$ implying $
E_\a  = O(L^{-2})
$ and $
H_\a  = O(L^{-2})
.$

The evolution equations for a spatially homogeneous Maxwell field are given by \cite{vElUgg96} 
\begin{eqnarray}
\label{electric field evolution}
{\bf e_0} (E_\a) &=& -2 \Hubble E_\a + \sigma_{\a\b} E^\b - \Ncf_{\a\b} H^\b - \varepsilon_{\a\b\gamma} \Acf^\b H^\gamma 
 \\
\label{magnetic field evolution}
{\bf e_0} (H_\a) &=& - 2 \Hubble H_\a + \sigma_{\a\b} H^\b + \Ncf_{\a\b} E^\b + \varepsilon_{\a\b\gamma} \Acf^\b E^\gamma, 
\end{eqnarray}
with constraints $  2 \Acf_\a E^\a = 0 =  2 \Acf_\a H^\a$.
Introducing new variables $\mathcal{E}_\alpha = L^2 E_\alpha $ and $\mathcal{H}_\alpha = L^2 H_\alpha $, we can rewrite \eqref{electric field evolution} and \eqref{magnetic field evolution} as
\begin{eqnarray}
\label{rescaled electric field evolution}
{\bf \hat{e}_0} (\mathcal{E}_\a) &=&  \frac{\Sigma_{\a\b}}{L} \mathcal{E}^\b - \NCF_{\a\b} \mathcal{H}^\b - \varepsilon_{\a\b\gamma} \ACF^\b \mathcal{H}^\gamma,\\
\label{rescaled magnetic field evolution}
{\bf \hat{e}_0} (\mathcal{H}_\a) &=&  \frac{\Sigma_{\a\b}}{L} \mathcal{H}^\b + \NCF_{\a\b} \mathcal{E}^\b + \varepsilon_{\a\b\gamma} \ACF^\b \mathcal{E}^\gamma. 
\end{eqnarray}
It follows directly that ${\bf \hat{e}_0} (\mathcal{E}_\a)=O(1)$ and ${\bf \hat{e}_0} (\mathcal{H}_\a)=O(1)$, so that $\mathcal{E}_\alpha$ and $\mathcal{H}_\alpha$ have finite limits at $\tau_\infty$. Hence, all components of $L^4 T_{ab}$ are finite at $\tau_\infty$. Substituting this back into \eqref{rescaled electric field evolution} and \eqref{rescaled magnetic field evolution}, implies that ${\bf \hat{e}_0} (\mathcal{E}_\a)$ and ${\bf \hat{e}_0} (\mathcal{H}_\a)$ have finite limits too.

The derivatives of the Maxwell tensor and the energy momentum tensor are made up of derivatives of the Faraday tensor and connection coefficients. Hence, it can be deduced, from the estimates obtained above, that $L^3 \nabla_a F_{bc}$ is finite at $\tau_\infty$ and a similar result holds for the derivatives of the energy-momentum tensor \eqref{ola11}.
\subsubsection{Massless Einstein-Vlasov matter} 
These models derive from relativistic kinetic theory and the corresponding system of equations is a sub-class of the Einstein-Boltzmann system for the case of massless (ultra-relativistic) particles with no collision term. The DEC and SEC hold in this case so the previous estimates for general trace-free matter can be used.  In turn, estimates for the kinetic variables are derived from the Vlasov equation and are presented in subsection \ref{Vlasov-matter-sec} together with the massive case which is not trace-free. 
\subsection{Einstein-scalar field models}
\label{Section - massive scalar field}
The energy-momentum tensor for a massive scalar field model is given by
\begin{equation}
\label{Einstein-scalar}
T_{\mu\nu}= \nabla_\mu\phi\nabla_\nu\phi-\left(\frac{1}{2}\nabla_\rho\phi\nabla^\rho\phi+V(\phi)\right)g_{\mu\nu},
\end{equation}
with $\phi$ satisfying
\begin{equation}
\label{scalar field KG-equation}
\nabla^\mu\nabla_\mu\phi=V'(\phi),
\end{equation}
where $V$ is a sufficiently smooth potential and $V' = \textrm{d}V/\textrm{d}\phi$. 

We make the following additional assumptions about the potential (see also  \cite{Ren04b}\footnote{ The analysis in \cite{Ren04b} considers space-times with a vanishing cosmological constant and a scalar field whose potential has a strictly positive lower bound. The constant $V_1$ in Theorem 1 of \cite{Ren04b} is equivalent to the cosmological constant $\Lambda$ used in this article. Thus, the arguments employed in \cite{Ren04b} carry over directly to this analysis.}) 
\begin{enumerate}
	\item $V'$ is bounded whenever $V$ is bounded.
	\item $V''(\phi)\to V_2>0$, as $t \to \infty$.
	\item $L^3(2\lambda^2\phi-V'(\phi))$ is bounded.
\end{enumerate}
For a spatially homogeneous scalar field, the gradient $\nabla_\mu\phi$ is normal to the hypersurfaces of spatial homogeneity. Hence, $\nabla_a\phi = - \dot{\phi} n_a$, where $\dot{\phi} = {\bf e}_0(\phi)$. 
In this case, given the identification of ${\bf n}$ with ${\bf u}$, the energy-momentum tensor \eqref{Einstein-scalar} can be written in the form of an aligned perfect fluid with
\begin{equation}
  \rho = \frac{1}{2}\dot{\phi}^2 + V(\phi) ,
  \quad
  p = \frac{1}{2}\dot{\phi}^2 - V(\phi) ,
  \quad q_a=0,
  \quad \pi_{ab}=0,
\end{equation}
while \eqref{scalar field KG-equation} gives rise to
\begin{equation}
  \ddot{\phi} + 3 \Hubble \dot{\phi} + V'(\phi) = 0.
\end{equation}
The SEC and DEC are satisfied if and only if $\dot{\phi}^2 \ge V(\phi)  \ge 0$.
If Assumption \ref{Assumptions Wald+} holds, then Proposition \ref{prop-decay1} tells us that $L^2(\rho+p)=L^2\dot{\phi}^2$ and $L^2(\rho-p)=L^2 V(\phi)$ must be bounded. Hence, $L\dot{\phi}={\bf \hat e}_0 (\phi)$ is bounded and thus $\phi$ has a finite limit $\phi_\infty :=\phi(\tau_\infty)$.

Equation \eqref{ids1SH} takes the form ${\bf e}_0(\rho) = -3\Hubble \dot{\phi}^2 \le 0$ and, hence, $\rho \ge V(\phi)\ge 0$ is a decreasing function in $t$. It follows that $V(\phi)$ is a non-negative bounded function of $\phi$. Hence, by our assumptions, $V'(\phi)$ is also bounded. Following the argument of \cite{Ren04b} this implies that $V'(\phi_\infty)=0$.
 
 Without loss of generality, we redefine the scalar field by $\phi \mapsto \phi-\phi_\infty$, as the energy-momentum tensor \eqref{Einstein-scalar} and the wave equation \eqref{scalar field KG-equation} simultaneously keep their form if $V(\phi) \mapsto V(\phi-\phi_\infty)$. So, from here onwards in this section, we will have $\phi_\infty = 0 $, $V(0)=0$ and $V'(0)=0$. 

The assumption $V''(0)=V_2>0$ implies that, given $\epsilon>0$, the following estimates hold for sufficiently small $\phi$, i.e. sufficiently late times \cite{Ren04b}:
\begin{equation}
\label{phi quadratic bounds on V}
(V_2-\epsilon)\phi^2 \le V(\phi) \le (V_2+\epsilon)\phi^2 \qquad \textmd{and} \qquad 
(V_2-\epsilon)\phi^2 \le \phi V'(\phi) \le (V_2+\epsilon)\phi^2. 
\end{equation}
In particular, the first inequality implies that the rescaled scalar field $\hat \phi:=L\phi $ is bounded, which, in turn, implies that $\phi$ decays exponentially. Moreover, ${\bf e}_0(\hat \phi) = \Hubble \hat \phi+ L\dot{\phi} $ is bounded and we observe that
\begin{eqnarray}
  \frac{\textmd{d}^2\hat \phi}{\textmd{d}\tau^2} := {\bf \hat e}_0 ({\bf \hat e}_0 (\hat \phi)) &=& L^3 (\ddot{\phi} + 3 \Hubble \dot{\phi}) + L^3(\dot{\Hubble} + 2 \Hubble^2)\phi
  \nonumber \\
  &=& L^3 (2 \lambda^2\phi - V'(\phi))
  + \psi \left(
  \frac{1}{6}(\rho -3p)L^2 - \frac{\Sigma^2}{6L^2}- \frac{\hat{R}}{6}
  \right) \nonumber \\
  \label{2nd tau derivative of psi}
&=& L^3 (2 \lambda^2\phi - V'(\phi)) +O(1).
\end{eqnarray}
By assumption 3 of this section, $\textmd{d}^2\hat \phi/{\textmd{d}}\tau^2$ is bounded and hence $\textmd{d}\hat \phi / \textmd{d}\tau$ and $\hat \phi$ have finite limits at $\tau_\infty$ (then $\phi_\infty=0$).

We now take a closer look at assumption 3. Writing $V(\phi)= \sum_n  \frac{V_n}{n!}\phi^n$ we have
\begin{equation}
L^3 (2 \lambda^2\phi - V'(\phi)) = (2\lambda^2\hat \phi-V_2\hat \phi) L^2 - \half V_3 L \hat \phi^2 - \frac{V_4}{3!}\hat \phi^3 + h.o.t.,
\end{equation}
where $V_0=V_1=0$ due to $V(0)=V'(0)=0$. 
Since $\hat \phi$ is bounded, then assumption 3 is satisfied if $V_2=2\lambda^2 $ and $V_3=0$, and so setting $m^2=2\lambda^2$ the potential $V(\phi)$ takes the form
\begin{equation}
\label{scalar-field-potential}
V(\phi) = \frac{1}{2}m^2 \phi^2 + \phi^4 U(\phi),
\end{equation}
where $U$ is a smooth function.
This matches the conditions for the potential used in \cite{Fri-massive-scalar}, as $V_2= m^2 = \frac{2}{3}\Lambda$. We will comment further on the role of \eqref{scalar-field-potential} in Section \ref{Section - Conformal regularity}.
%
\subsection{Einstein-Vlasov matter}
\label{Vlasov-matter-sec}
Vlasov models are a class of kinetic models with collisionless particles. Let $T^\star M$ be the cotangent bundle of the physical manifold $M$ and $f\in C^1(T^\star M; \R)$ the particle distribution function. Let $x\in M$. The particle distribution with momentum ${\bf p}\in T^\star_x M$, for particles with mass $m$, is a non-negative function defined on
$$
{\cal P}= \{(x,{\bf p})\}\in T^\star M: {\bf g}^{-1}_x({\bf p},{\bf p})=-m^2,~~\text{${\bf p}$ is future pointing}\}.
$$
The Einstein-Vlasov energy-momentum tensor at $x$ is given by \cite{Ren04-vlasov}
\begin{equation}
\label{Einstein-Vlasov}
T_{\mu\nu}(x)= \int_{{\cal P}_x} f(x, p) p_\mu p_\nu \frac{\sqrt{\vert \text{det}(g^{\mu\nu}) \vert}}{g^{0\mu}p_\mu} d p_1dp_2dp_3,
\end{equation}
An important aspect regarding the application of our results is that \eqref{Einstein-Vlasov} satisfies the DEC and SEC, see e.g. \cite{Ren04-vlasov, Rendall-Tod}. 
The evolution equation for $f$ is given by the {\em Vlasov equation} 
\begin{equation}
\label{Vlasov-equation}
g^{\mu\nu}p_\mu \partial_{x^\mu}f-\frac{1}{2}\partial_{x^\rho} g^{\mu\nu}p_\mu p_\nu \partial _{p_\rho}f=0 .
\end{equation}
Spatially homogeneous solutions of the Einstein-Vlasov system have, by definition, the property that both geometry and the phase space density of particles are invariant under the group action defining the Bianchi type. In the evolution equations $\partial_{x^\mu}$ is then $\partial_{t}$, but the Vlasov equation is still a PDE considering the terms in $\partial_{p_\rho}$.  
Under our Assumptions \ref{Assumptions Wald+}, Lee derived the estimate \cite{Lee}
\begin{equation}
\label{momentum-decay}
|p_\mu |=C+O(e^{-\lambda t}),
\end{equation}
where $C>0$ is constant, so $p_\mu$ has a finite limit as $t\to +\infty$.  

An important assumption in \cite{Joudioux, Lee}, which we also use here, is that $f$ has compact support. In particular, the initial data $ f(0, p_\mu)$ has compact support in $p_\mu$ and, as discussed in \cite{Lee}, since $f$ is constant along the characteristic curves of \eqref{Vlasov-equation},  then it is bounded by its initial data as $f(t, p_\mu)\le \sup{\{|f(0,p_\mu)|: \forall p_\mu\}}$.

For the massive case $m>0$, Proposition 6 in \cite{Lee} gives the decay estimates in an orthonormal frame
\begin{eqnarray}
\label{bounds for Vlasov}
\rho &=& O(e^{-3\lambda t}) \nonumber \\
q_\a  &=& O(e^{-4\lambda t}) \\
\pi_{\a\b} &=&O(e^{-5\lambda t}).\nonumber
\end{eqnarray}
Thus, massive Einstein-Vlasov satisfies our Proposition \ref{prop-decay2}, case 4, with $k=-2$. 

In turn, the massless case corresponds to trace-free matter, so the estimates in the preamble of Section \ref{tracefree} apply directly to give $\rho, q_\a, \pi_{\a\b} = O(e^{-4\lambda t}) $.
\subsection{Multiple matter sources}
\label{multiple-matter}
Above, we considered decay rates for space-times with a single matter source, treating each model separately. Now consider a space-time with multiple matter components ${\cal M}^I$, each with an energy-momentum tensor $T^I_{\mu\nu} $. We assume that $\nabla^\mu T^I_{\mu\nu} =0$ holds for each individual matter component ${\cal M}^I$ separately. Using \eqref{ids1SH} and \eqref{efe51SH}, this assumption implies that 
\begin{eqnarray}
\label{evol rho indiv}
{\bf e}_0({\rho^I})&=&- 3(\rho^I+p^I)\Hubble-\pi^I_{\a\b}\sigma^{\a\b} + 2 a^\a q^I_\a
\\
\label{evol q indiv}
{\bf e}_0 (q^I_\a)&=&
-4\Hubble q^I_\a - \sigma_{\a\b}q^{I\b} + 3a_\b\pi^{I \b}_\a +\pi^{I\gamma}_\b\varepsilon_{\gamma\a\delta}n^{\b\delta}.
\end{eqnarray}
The overall energy-momentum tensor \eqref{ola10} is given by 
$T_{\mu\nu} = \sum T^I_{\mu\nu} $ with $ \rho = \sum \rho^I, \,\, q_\a  = \sum q_\a^I \,\, \textmd{etc}$. The evolution and constraint equations use the components of $T_{\mu\nu} $ and the geometric variables.
Combining these equations with \eqref{evol rho indiv} and \eqref{evol q indiv} we can adapt the analysis in this section to find decay rates for space-times with multiple matter sources. As our previous analysis has shown, different matter models lead to different decay rates for the curvature and matter components. When multiple matter models are present, the overall decay rates derived in Proposition \ref{prop-decay2} depend on the behaviour of the model with the slowest decay rate.
%
\section{Asymptotic simplicity for Bianchi space-times}
\label{Section - Conformal regularity}
In order to introduce the notion of {\em asymptotic simplicity}, Penrose \cite{Pen64,Pen64b,Pen65} considered that the {\em physical space-time} $(M,g)$  (where the EFE are satisfied) is conformally embedded into a larger regular {\em unphysical space-time} manifold $(\hat{M}, \hat{g})$. Moreover, on $M$ the two metrics are conformally related by $\hat{g}_{\mu\nu}=\Omega^2 g_{\mu\nu}$, where the conformal factor $\Omega$ is a boundary defining function of $M$ in $\hat{M}$, i.e. $\Omega>0$ on $M$ and on the (conformal) boundary $\Scri$ of $M$ in $\hat{M}$ we have $\Omega=0$ and $\mathrm{d}\Omega \neq 0$.
For more details see e.g. \cite{JVKbook}.

For $\Lambda > 0$, the set $\Omega=0$ (denoted by $\Scri$) is a spacelike 3-dimensional hypersurface. $\Scri$ is referred to as null infinity since it contains the idealised endpoints of infinitely extendable null geodesics of the metric $g$. Points on $\Scri$ are infinitely far away for any observer in $(M,g)$ and thus $\Scri$  represents the  infinity of $(M,g)$. In our analysis we are only concerned with space-times whose null geodesics acquire an endpoint on $\Scri^+$, that is, to the future. Such space-times are referred to as future asymptotically simple. For sufficiently smooth conformal extension $(\hat{M}, \hat{g})$
the matter content and the Weyl curvature of $(M,g)$ vanish asymptotically. 

In \cite{Wal83} Wald discusses the existence of a large class of Bianchi space-times that locally asymptotically approach de Sitter. However, for asymptotic simplicity one needs to verify that the decay rates of the curvatures of the Bianchi space-time are sufficiently strong to permit the existence of a conformal extension $(\hat{M}, \hat{g})$ which is regular on $\Scri^+$. That this is a non-trivial matter is shown by the example of the Nariai space-time \cite{Bey09a} (see also \cite{Fri-review}). In that case, although the conformal metric extends smoothly to the conformal boundary, the function $\Omega$ only extends continuously with $\Omega=0$ but $d\Omega$ is divergent. Moreover, the scalar quantity  $\hat C_{\mu\nu\lambda\rho}\hat C^{\mu\nu\lambda\rho}$ does not vanish as required when one approaches the conformal boundary.

In this section we will address the following questions:
\\\\
\textit{(i) Can we conformally embed the physical Bianchi space-time $(M,g)$ with $\Lambda>0$ into a regular unphysical space-time $(\hat{M}, \hat{g})$ with finite connection coefficients and curvature, where future null infinity $\Scri^+$ is represented by the  spacelike hypersurface ${\hat{\cal{S}}}_\infty =\{\tau=\tau_\infty \}$? 
\\\\
(ii) Are the Bianchi space-times analysed by Wald \cite{Wal83} asymptotically simple?} 
\\\\
These questions will be answered affirmatively in Theorem \ref{prop-future-asymp-simple} for spacetimes where the anisotropic stress decays sufficiently fast asymptotically. We will then derive asymptotic constraints at  $\Scri^+$ in terms of conformal variables in Proposition \ref{prop-constraints at Scri} and then show that they are satisfied for certain matter fields in Proposition \ref{thm-constraints}.
\subsection{Conformally rescaled geometric quantities}
\label{subsec-conformal-variables}
%
In this subsection, we simply recall the relationship between the variables in the physical and unphysical space-times. The unphysical variables are needed both for our  analysis and regularity of the CEFE later on. For more details, see \cite{FraLR, Fri-review, JVKbook}.

Variables defined in terms of $\hat{g}_{\mu\nu}$ will be denoted with a hat '$~\hat{}~$'.
If ${\bf e}_a$ is a $g$-orthonormal frame, then the frame 
$${\bf \hat{e}}_a = \Omega^{-1} {\bf e}_a$$ 
is a $\hat{g}$-orthonormal frame. We adopt the convention that frame components of quantities related to $\hat{g}_{\mu\nu}$ are evaluated in the frame ${\bf \hat{e}}_a$. Where ambiguity could arise we will specify the frame explicitly.
The frame connection coefficients corresponding to the two frames
are related by 
\begin{equation}
\label{relation of frame connection coefficients}
\hat{\Gamma}\tensor{}{c}{ab}  = \Omega^{-1}
(\Gamma\tensor{}{c}{ab} + \delta\tensor{}{c}{a}\Upsilon_b 
- g_{ab}g^{cd}\Upsilon_d),\qquad \textmd{where} \qquad \Upsilon_a={\bf e}_a (\log \Omega)=\frac{{\bf e}_a \Omega}{\Omega}.
\end{equation}
Hence, for a fluid flow along ${\bf e}_0={\bf u}$, the frame components of the kinematic quantities given in Section \ref{setup} are related to their hatted counterparts by
\begin{eqnarray}
\label{transformation of  individual connection coefficients}
&&\hat{u}_a = \Omega^{-1} u_a, \quad \quad\quad\quad\quad
\hat{\sigma}_{ab} = \Omega^{-1} \sigma_{ab}, \quad
\nonumber
\\
&&\hat{H} = \Omega^{-1} (H - \Upsilon_a u^a) , \quad 
\hat{A}^a = \Omega^{-1} (A^a + (g^{ab}+u^a u^b) \Upsilon_b).
\end{eqnarray}
In the rescaled Bianchi space-time the surfaces of homogeneity are normal to ${\bf \hat e}_0={\bf \hat n}=L{\bf n}=L{\bf e}_0 $. 
They induce an unphysical metric $\hat{h}_{\mu\nu}=\hat{g}_{\mu\nu}+\hat{n}_\mu \hat{n}_\nu$, with Levi-Civita connection $\hat{D}_\a$ and extrinsic curvature $\hat{\chi}_{ab} $ of the homogeneity hypersurfaces embedded into $(\hat{M}, \hat{g}) $. The same notation is adapted for the spacelike hypersurface ${\hat{\cal{S}}}_\infty=\{\tau = \tau_\infty\}$. 

The Ricci curvature $R_{\mu\nu}$ can be expressed in terms of the Schouten tensor 
$$P_{\mu\nu}=\frac{1}{2}R_{\mu\nu}-\frac{1}{12}Rg_{\mu\nu}$$
and frame components of the Schouten tensors of $\hat{g}$ and $g$ are related by
\begin{equation}
\label{Schouten transformation}
\hat{P}_{ab} = \Omega^{-2}\left[ P_{ab} - {\bf e}_a(\Upsilon_b) + \Gamma\tensor{}{c}{ab}\Upsilon_c + \Upsilon_a \Upsilon_b - \frac{1}{2} \Upsilon_c \Upsilon_d g^{cd} g_{ab}. \right].
\end{equation}
The Weyl tensor itself is conformally invariant as are its electric and magnetic parts. However, their respective frame components are related by
\begin{equation}
\label{unphysical EMWeyl}
\hat{C}\tensor{ab}{c}{d}=\Omega^{-2}C\tensor{ab}{c}{d}, \quad 
\hat{E}_{ab} = \Omega^{-2}E_{ab}, \quad 
\hat{H}_{ab} = \Omega^{-2}H_{ab}.
\end{equation}
So, the Riemann curvature $\hat{R}\tensor{ab}{c}{d}$ of the unphysical metric $\hat{g}$ is finite if the frame components \eqref{Schouten transformation} and \eqref{unphysical EMWeyl} have finite values.

For the CEFE approach, we need the conformally rescaled Weyl tensor associated to $\hat{g}$ and defined as 
$$d\tensor{\mu\nu}{\lambda}{\rho}:= \Omega^{-1}C\tensor{\mu\nu}{\lambda}{\rho},$$
while the ${\bf \hat{e}}_a$ frame components of its electric and magnetic parts are given by 
\begin{equation}
\label{rescaled confomral EMWeyl}
{\cal{E}}_{ab} = \Omega^{-3}E_{ab}, \quad 
{\cal{H}}_{ab} = \Omega^{-3}H_{ab}.
\end{equation}
Furthermore, one defines 
\begin{equation}
\label{def-d-and-s}
d_a :=  {\bf \hat e}_a( \Omega)~~{\text {and}}\quad s 
:=-\frac{1}{4} \hat{g}^{ab} \hat{\nabla}_{a} d_{b}
=-\frac{1}{4} \hat{g}^{ab} \hat{\nabla}_{a} \hat{\nabla}_{b}\Omega,
\end{equation}
where $\hat\nabla$ is the covariant derivative of metric $\hat g$.
%
\subsection{Future asymptotically simple Bianchi space-times}
\label{Sec:Penrose approach}
In order to embed the FLRW space-time into the conformally flat Einstein cylinder, one chooses the inverse of the scale function $a(t)$ as the conformal factor $\Omega$. The fact that in our setting here $a(t)$ corresponds to the length scale $L(t)$, motivates the following choice of conformal factor for our analysis
\begin{equation}
\label{Conformal factor}
\Omega(t) = \frac{1}{L(t)}.
\end{equation}
Note that the choice \eqref{Conformal factor} matches the rescaling of the physical variables introduced in Section \ref{Section unphysical variables}, which was introduced purely to simplify the evolution equations so that we could obtain improved estimates. 

Recall from our earlier analysis, that Assumption \ref{Assumptions Wald+} implies that $L=O(e^{\lambda t})$ as $t \to \infty$. Hence, $\Omega=O(e^{-\lambda t})$ vanishes as $t \to \infty$ as required for a conformal embedding.
Since $\Omega$ is a function of $t$ only, it  follows that $\Omega$ is constant on each surface of homogeneity $\cal S$ and hence the property of spatial homogeneity (and possibly isotropy) is preserved in the unphysical space-time. Moreover, \eqref{transformation of  individual connection coefficients} implies  $\hat{H}=0$ so that an observer in the unphysical space-time $(\hat{M}, \hat{g})$ no longer observes any expansion.

Since $L(t)$ is a strictly increasing function of the conformal time $\tau$ defined in \eqref{conformal time}, it maps $[0,\infty)$ to the finite time interval $[0,\tau_\infty)$ and we have $\Omega(\tau_\infty)=0$. Hence, the conformal rescaling \eqref{Conformal factor} compactifies the Bianchi space-times in the time direction. In turn, the foliation by $t$ is mapped to a foliation by $\tau$ so that each surface of constant $\tau$ is a surface of homogeneity. Below we will show that the spacelike hypersurface ${\hat{\cal{S}}}_\infty=\{\tau = \tau_\infty\}$ represents $\Scri^+$.

\begin{remark} Any covariant derivative using $\hat{\nabla}$, can be expressed as derivatives along the frame vectors plus terms involving the finite frame connection coefficients $ \hat{\Gamma}\tensor{}{c}{ab}$. Due to our gauge choice of ${\bf {e}}_a$, respectively ${\bf \hat{e}}_a$, and the fact that \eqref{Conformal factor} preserves spatial homogeneity, we only require derivatives along ${\bf \hat{e}}_0$. Hence, if the frame components of a spatially homogeneous quantity $\bf Q$ have a finite limit at $\tau_\infty$, then the covariant derivative $\bf \hat{\nabla} {\bf Q}$ is finite as long as ${\bf \hat{e}}_0({\bf Q})$ is finite.
\end{remark}

If $\Omega$ only changes in the ${\bf e}_0$-direction, then we have 
\begin{equation}
\label{Upsilon-Bianchi}
\langle \Upsilon, {\bf e}_a \rangle =  \langle \mathrm{d}\Omega, \hat{{\bf e}}_a \rangle  = -\Hubble \delta_{~a}^0 \neq 0 \qquad \textmd{and}\qquad
\hat{\Gamma}\tensor{}{c}{ab}  = L (\Gamma\tensor{}{c}{ab} - H\delta\tensor{}{c}{a}\delta\tensor{}{0}{b}
+ H g_{ab}g^{c0}).
\end{equation}
Hence, $\Upsilon_\a=0$, $\mathrm{d}\Omega \neq 0$ and $\hat{\Gamma}\tensor{}{\g}{\a\b}  = \Omega^{-1}\Gamma\tensor{}{\g}{\a\b}$. 
Moreover, \eqref{transformation of  individual connection coefficients} implies
\begin{eqnarray}
\hat{\Hubble}=0, \quad \hat{L}=1, \quad
\label{unphysical-frame-connection-components}
\hat{\sigma}_{ab} = L \sigma_{ab},\quad 
\hat{n}_{ab} = L n_{ab}, \quad \hat{a}_{a} = L a_{a}, \nonumber
\end{eqnarray}
together with $\hat{A}^b = L A^b=0$ and $\hat{\omega}_{\a\b} = L \omega_{\a\b}=0$.
Furthermore, it follows that 
$$^{(3)}\hspace{-0.05cm}\hat R\,=\,L^2\, ^{(3)}\hspace{-0.05cm}R ,\qquad ^{(3)}\hspace{-0.05cm}\hat S_{ab}\,=\, L^2 \, ^{(3)}\hspace{-0.05cm}S_{ab} \quad \textmd{and}\quad \hat Y_{ab}\,=L^3 Y_{ab},$$ 
where $\hat{Y}_{ab}$ is the Cotton-York tensor given in \eqref{3-Cotton-York}. 
For the subsequent calculations define 
\begin{equation}
\label{Xi}
\Xi_{ab}:= L^{-2} (\tfrac{1}{2}\pi_{ab}-H \sigma_{ab})=\frac{\Pi_{ab}}{2L}-H \Sigma_{ab},
\end{equation}
where
$$\Pi_{ab}:=L^3\pi_{ab},~~~~{\text {and}}~~~\Sigma_{ab}:=L^2\sigma_{ab}.$$
\begin{theorem}{\bf (Asymptotically simple Bianchi space-times)}
\label{prop-future-asymp-simple} 
\\
Suppose Assumption 1 holds and $\Xi_{ab}$ has a finite limit as $\tau\to\tau_\infty$.
Then the Bianchi space-times can be conformally extended to ${\hat{\cal{S}}}_\infty$ such that
\begin{enumerate}
	\item the connection coefficients and the Riemann curvature associated to $\hat{g}_{\mu\nu}=L^{-2} g_{\mu\nu}$ are finite on ${\hat{\cal{S}}}_\infty$. In fact, this result holds for $\pi_{ab}=O(L^{-(2+\epsilon)}) $, with $\epsilon>0$;
	\item the congruence formed by $\hat{\bf e}_0$ and parameterised by $\tau$ is geodesic and free of conjugate points at ${\hat{\cal{S}}}_\infty$ with respect to $\hat{g}$;
	\item the space-times are future asymptotically simple and ${\hat{\cal{S}}}_\infty$ represents future null infinity $\Scri^+$.
\end{enumerate}  
\end{theorem}
\begin{proof}
1) In Section \ref{Section unphysical variables}, we showed that all the frame connection coefficients $\hat{\Gamma}\tensor{}{c}{ab}$ have a finite limit at $\tau_\infty$. 
Hence $\hat{\chi}_{ab}$ 
 decays like $e^{-\lambda t} $ and vanishes on ${\hat{\cal{S}}}_\infty$.

For the unphysical Schouten curvature, we combine \eqref{Schouten transformation}, \eqref{Conformal factor} and \eqref{Upsilon-Bianchi} to get
\begin{eqnarray}
\hat{P}_{ab} &=& L^2 (P_{ab} + (\dot{\Hubble}+\Hubble^2) \delta_{~a}^0 \delta_{~b}^0 - \Hubble \Gamma\tensor{}{0}{ab} + \tfrac{1}{2}\Hubble^2 g_{ab})
\nonumber \\
\label{unphysical Schouten expansion}
&=& \frac{\Pi_{ab}}{2L}-H \Sigma_{ab} +\frac{1}{L} Q_{(a}\delta^0_{~b)} + 
\left( \frac{^{(3)}\hspace{-0.05cm}\hat{R}}{6}- \frac{\Sigma^2}{L^2} \right) \delta_{~a}^0 \delta_{~b}^0
+\left(  \frac{^{(3)}\hspace{-0.05cm}\hat{R}}{12}- \frac{ \Sigma^2}{6L^2} 
\right)g_{ab},
\end{eqnarray}
where we have used \eqref{unphysical variables defined}, \eqref{Schouten in terms of matter}, \eqref{efe0SH} and \eqref{efe6SH} in the second step. 
We observe that $\hat{P}_{ab}$ does not dependent explicitly on density and pressure. Thus if $\Xi_{ab}$ has a finite limit at $\tau_\infty$, then the unphysical Schouten tensor has a finite limit on  ${\hat{\cal{S}}}_\infty$. Moreover, that condition guarantees that the unphysical electric Weyl tensor $\hat{E}_{ab}=L^2 E_{ab}$ also has a finite limit at $\tau_\infty$, while the unphysical magnetic Weyl tensor $\hat{H}_{ab}=L^2 H_{ab}$ vanishes by Proposition \ref{prop-decay1}. Moreover, by Proposition \ref{prop-decay2}, the asymptotic finite condition on $\Xi_{ab}$ is satisfied if $\Pi_{ab}=O(L^{k}) $ with $k<1$, which is equivalent to $\pi_{ab}=O(L^{-(2+\epsilon)}) $, with $\epsilon>0$.

2) Since $\hat{A}^\mu=0$ the curves are $\hat{g}$-geodesics. Let $\eta^\mu=z \epsilon^\mu$ denote a space-like Jacobi field of length $z(\tau)$, so that $\epsilon^\mu$ is a space-like unit vector. Using $\hat{H}=0$ and setting $z(0)=1$ it follows that 
\begin{equation}
    \frac{\mathrm{d}z}{\mathrm{d}\tau} = \hat{\sigma}(\epsilon,\epsilon) \, z \ge - C \Vert  \sigma \Vert z \implies 
    z(\tau) \ge e^{- C \Vert  \sigma \Vert\tau} \quad \mathrm{for}\quad \tau \le \tau_\infty,
\end{equation}
where we recall that $\Vert  \sigma \Vert$ denotes the euclidean norm of a matrix with entries $\sigma_{\alpha\beta}$.
Since $z>0$ the congruence is free of conjugate points up to at least ${\hat{\cal{S}}}_\infty$.

3) The manifold $\hat{M}=M \cup {\hat{\cal{S}}}_\infty$ is an conformal extension of the Bianchi space-time $(M,g)$. By our definition $\Omega=0$ on ${\hat{\cal{S}}}_\infty$ while \eqref{Hubble lower and upper bound} and \eqref{Upsilon-Bianchi} show that $\mathrm{d}\Omega\ne 0$ on ${\hat{\cal{S}}}_\infty$. Thus the final statement follows.
\end{proof}
\begin{remark}
	As discussed earlier, the Nariai space-time does not admit a regular conformal extension as $\mathrm{d}\Omega = 0$. It follows that the Nariai space-time can not satisfy both Assumption \ref{Assumptions Wald+} and Assumption \ref{Assumption for Pi} with $k\le-1$ (through Proposition \ref{prop-decay2}), as those assumptions imply that $\mathrm{d}\Omega \neq 0$ and that the aforementioned scalar quantity $\hat C_{\mu\nu\lambda\rho}\hat C^{\mu\nu\lambda\rho}$ must vanish as $t\to \infty$.
\end{remark}
\begin{remark}
The existence of $\scri^+$ in the case of Vlasov matter or pure radiation ($\c=\tfrac{4}{3}$) in Bianchi spacetimes with $\Lambda>0$ was also established by Tod  \cite{Tod-2007}, for initial data given at an isotropic singularity.  
\end{remark}
\subsection{Conformal hair}
\label{sec:conformal hair}
It is a well established fact that the  Bianchi type of $(M,g)$ is preserved through out the entire evolution \cite{EllMac69, Ellis-Elst} in the physical space-time. Since our conformal factor \eqref{Conformal factor} only depends on $t$, then it is constant on each surface of homogeneity. As a consequence, the Killing vector fields $\bm \xi_\a$ of $(M,g)$ are also Killing vector fields $\bm \xi_\a$ of $(\hat{M},\hat{g})$ and the unit normal vector $ \hat{\bm n}$ satisfies $[ \hat{\bm n},\bm \xi_\a ]=0$ again. Thus the structure constants  $\hat{C}^\c_{\a\b}$ for $(\hat{M},\hat{g})$ are equal to $C^\c_{\a\b}$ and can be extended onto ${\hat{\cal{S}}}_\infty$. Moreover, they are equivalent to spatial commutation functions defined by $\ACF_\a $ and $\NCF_{\a\b} $. 
Hence, as a corollary of Theorem \ref{prop-future-asymp-simple},  we can conclude:
\begin{corollary}{\bf (Conformal hair)} \label{conformal-hair}
\\
While the functions $a_\b$ and $n_{\a\b}$ associated to the physical Bianchi space-times $(M,g)$ vanish as $t\to \infty$ (also known as a no-hair property), their conformal counterparts  $\ACF_\a $ and $\NCF_{\a\b} $ have a well-defined non-zero limit on $\Scri^+$ and preserve information of the Bianchi type of $(M,g)$ at $\Scri^+$. Thus, the space-times exhibit conformal hair. 
\end{corollary}

\subsection{Asymptotic identities and constraints at conformal infinity}
\label{sec:asymptotic identities}

In Theorem \ref{prop-future-asymp-simple} we established conditions under which Bianchi spacetimes are asymptotically simple. We now take a closer look at the asymptotic identities that the geometric conformal variables
$$d_{a},~~ s,~~ \hat{\Gamma}\tensor{}{c}{ab},~~ \hat{P}_{ab}~~{\text {and}},~~ d_{abcd}$$
satisfy at  $\Scri^+$. We introduce the following notation: $A \circeq B $ if the quantities $A$ and $B$ have the same value on $\Scri^+$.

Interestingly, and quite generally, we are able to show the following asymptotic identities:
\begin{proposition}{\bf (Asymptotic identities)}
	\label{prop-constraints at Scri}
	~
	\begin{enumerate}
		\item Suppose Assumption 1 holds. Then, all frame connection coefficients $\hat{\Gamma}\tensor{}{c}{ab} $ are finite at $\Scri^+$ and
		\begin{equation}
		\label{constraints at Scri 0}
		d_a \circeq -\lambda \delta_{~a}^0, \quad s \circeq 0, \quad \hat{\chi}_{\a\b} \circeq 0, \quad \hat{P}_{00} \circeq \frac{^{(3)}\hspace{-0.05cm}\hat{R}}{6}.
		\end{equation}
		\item Suppose, in addition, that Assumption 2 holds. 
		\begin{enumerate}
			\item If $k<1$, then
			\begin{equation}
			\label{constraints at Scri 1}
			  \hat S_{\a\b}  \circeq -\lambda\Sigma_{\a\b} , \quad 
			  \hat{Y}_{\a\b} \circeq -\lambda{\cal{H}}_{\a\b}, \quad
			  \hat{P}_{\a\b} \circeq \hat{p}_{\a\b}, \quad \hat{P}_{0\b} \circeq 0,
			\end{equation}
			where $\hat{p}_{\a\b} =$$^{(3)}\hspace{-0.05cm}\hat{S}_{\a\b} + \tfrac{1}{12} ^{(3)}\hspace{-0.05cm}\hat{R}\delta_{\a\b}$ is the 3-dimensional Schouten tensor of $\hat{h}_{\mu\nu}$.
			\item If $k<0$, then 
			\begin{equation}
			\label{constraints at Scri 2}
				\varepsilon^{\a\b\gamma} \hat n_{\gamma\delta}\Sigma_{~\b}^\delta \circeq 3\Sigma^{\b\a} \hat a_\b.
			\end{equation}
			\item If $k\le-1$ or $\pi_{\a\b}=0$, then
			\begin{equation}
			\label{constraints at Scri 3}
				{\cal{E}}_{\a\b} \circeq L(\lambda\Sigma_{\a\b} + \hat S_{\a\b}) , \quad			
				\hat{D}^\a {\cal{E}}_{\a\b} \circeq -\lambda L Q_\b.
			\end{equation}
		\end{enumerate}  
	\end{enumerate}    	
\end{proposition}
\proof ~
\begin{enumerate}
	\item We already established earlier that the frame connection coefficients $\hat{\Gamma}\tensor{}{c}{ab} $ are finite at $\Scri^+$. In particular, using
	\eqref{unphysical-frame-connection-components}, this gives 
	$ \hat{\chi}_{\a\b} = \hat{\Hubble}\delta_{\a\b}+\hat{\sigma}_{\a\b}=\Sigma_{\a\b}/L \circeq 0$.
	Equation \eqref{def-d-and-s} directly implies $d_a=-\Hubble \delta_{~a}^0$ and hence the unphysical Laplacian of $\Omega$ can be rewritten as
	\begin{eqnarray}
	-4 s = \hat{g}^{ab} \hat{\nabla}_{a} \hat{d}_{b}  
	= \hat{e}_{0} (\Hubble) + \hat{g}^{ab}  \hat{\Gamma}\tensor{a}{0}{b} \Hubble 
	= L \dot{\Hubble} =O(L^{-1}) \quad \implies \quad s \circeq 0. \nonumber
	\end{eqnarray}
	The last limit of \eqref{constraints at Scri 0} arises directly from \eqref{unphysical Schouten expansion}.
	\item
	\begin{enumerate}
		\item In the proof of Proposition \ref{prop-decay2}, we observed that for $k<1$ we get  $\lambda\Sigma_{\a\b} + \hat S_{\a\b} = O(L^{k-1})$. Thus, $\hat S_{\a\b}  \circeq -\lambda\Sigma_{\a\b}$ follows immediately. If one replaces $\hat{S}_{\a\b}$ by $\Sigma_{\a\b} $ in \eqref{ola28SHR}, then one obtains \eqref{3-Cotton-York} and, hence, $\hat S_{\a\b}  \circeq -\lambda\Sigma_{\a\b}$ leads to $\hat{Y}_{\a\b} \circeq -\lambda{\cal{H}}_{\a\b}$.
		The identities for $\hat{P}_{\a\b}$ and $\hat{P}_{0\b}$ follow directly from \eqref{unphysical Schouten expansion} and the known decay rates.
		\item For $k<0$, we established in Proposition \ref{prop-decay2} that $Q_\a$ decays to zero, i.e. $Q_\a \circeq 0 $. Hence, \eqref{constraints at Scri 2} is a consequence of \eqref{efe1SHR}.
		\item The first equation of \eqref{constraints at Scri 3} follows from the known decay rates being applied to \eqref{ola26SHR}. For the second equation, we observe that 
		\begin{eqnarray*}
		\hat{D}^\a {\cal{E}}_{\a\b} &=& -3 \ACF^\a {\cal{E}}_{\a\b} 
		- \varepsilon_{\a\b\delta} {\cal{E}}^{\a\epsilon} \NCF\tensor{\epsilon}{\delta}{} \nonumber\\
		&\circeq& -3L \ACF^\a (\lambda\Sigma_{\a\b} + \hat S_{\a\b}) 
		-L \varepsilon_{\a\b\delta} (\lambda\Sigma^{\a\epsilon} + \hat S^{\a\epsilon}) \NCF\tensor{\epsilon}{\delta}{} \nonumber\\
		&\circeq& -\lambda L Q_\b +Z_\b,
		\end{eqnarray*}
		where $Z_\b=-3L \ACF^\a  \hat S_{\a\b}
		-L \varepsilon_{\a\b\delta} \hat S^{\a\epsilon} \NCF\tensor{\epsilon}{\delta}{} $. It is thus sufficient to show that $Z_\b$ vanishes identically. This can be done by expanding the right hand side in terms of $\ACF_\a$ and $\NCF_{\a\b}$ using \eqref{efe70SHR}. 
		Note that some terms vanish due to the constraints $\ACF^\a \NCF_{\a\b}=0$ and $\hat a^\alpha \hat a_\beta \varepsilon^{\beta\delta}_{~~\alpha}=0$. 
		Moreover, we observe that $\NCF_{\a\gamma}\NCF\tensor{}{\gamma}{\b}$ and $\NCF_{\a\gamma}\NCF\tensor{}{\gamma}{\delta}\NCF\tensor{}{\delta}{\b}$ are symmetric in $\a\b$. Hence, any contraction with $\varepsilon_{\a\b\gamma} $ will vanish. Finally, using the identities
	\begin{center}
			$
		\varepsilon_{\a\b\gamma} \varepsilon^{\a\delta\epsilon} 
		= 2 \delta^{[\delta}_{~\b} \delta^{\epsilon]}_{~\gamma},~~\text{and}~ 
		\varepsilon_{\a\b\gamma}\varepsilon^{\delta\epsilon\kappa} 
		= 6 \delta^{[\delta}_{~\a} \delta^{\epsilon}_{~\b} \delta^{\kappa]}_{~\gamma},
		$
	\end{center}
		we can remove all occurrences of $\varepsilon_{\a\b\gamma}$, leaving an expression in terms of $\ACF_\a$ and $\NCF_{\a\b}$. The remaining terms either cancel directly or vanish due to the constraint $\ACF^\a \NCF_{\a\b}=0$ once more. Thus, one finally derives that $Z_\b$ vanishes identically and, hence, \eqref{constraints at Scri 3} holds.
		\endproof
	\end{enumerate}
\end{enumerate}
In particular, we can show that the various matter fields that have been analysed in Section \ref{Section - matter models} satisfy the asymptotic identities of Proposition \ref{prop-constraints at Scri}: 

\begin{proposition} {\bf (Matter models satisfying the asymptotic identities)}\label{thm-constraints}
\\
Suppose we have a Bianchi space-time satisfying Assumption 1, whose matter model corresponds to either:
	\begin{itemize}
		\item Tracefree matter\footnote{This contains all cases of Section \ref{tracefree} including Maxwell fields and massless Vlasov matter.},
		\item Massive scalar field with potential given by \eqref{scalar-field-potential}, 
		\item Aligned dust,
		\item Aligned perfect fluids with $\gamma \in [\tfrac{2}{3},2]$, 
		\item Massive Vlasov matter, with $f$ being $C^1$ and having compact support.
		\item Elastic matter with \eqref{elastic} and $\gamma \ge \frac{4}{3}$, 
		\item Viscous fluids with \eqref{viscous-T}, \eqref{viscous-example} and $\gamma \ge \frac{4}{3}$.
	\end{itemize} 
	Then, each of these space-times satisfies the identities \eqref{constraints at Scri 0}-\eqref{constraints at Scri 3} at conformal infinity. 
\end{proposition}
\proof From the analysis in Section \ref{Section - matter models}, we know that all cases satisfy Proposition \ref{prop-constraints at Scri} with $k\le -1$ and, hence, for each case, the asymptotic identities \eqref{constraints at Scri 0}-\eqref{constraints at Scri 3} hold. 
\endproof
%
\section{Non-linear stability for Bianchi space-times}
\label{Sec:Friedrich approach}
The advantage of the conformal approach is that it allows studying the asymptotic structure of $(M,g)$ through a local analysis in $(\hat{M}, \hat{g})$ near the conformal boundary $\Scri$. However, when one tries to formulate a conformal version of the Cauchy problem to solve the EFE directly up to and including $\Scri$, one faces the problem that the unphysical Ricci curvature and other quantities pick up factors of $\Omega^{-1}$ so that the differential equations become singular as $\Omega \to 0$ at $\Scri$.
For vacuum space-times Friedrich \cite{Fri81a,Fri81b,Fri-vacuum,Fri-dS} resolved this problem by expanding the number of  variables considered as unknowns\footnote{The system used the conformal factor, its gradient and its Laplacians, an orthonormal frame, the unphysical connection coefficients, the unphysical Schouten curvature and the rescaled Weyl curvature as its variables.} in order to obtain a closed first order system of PDEs 
regular at $\Scri^+$, where $\Omega=0$. This  system is known as the {\em conformal Einstein field equations} (CEFE). In the meantime the CEFE approach has been extended to space-times containing electromagnetic fields and Yang-Mills fields \cite{Fri-EMYM, LueVal12a}, conformal scalar fields \cite{Huebner}, radiation fluids \cite{LueVal12b}), massive scalar fields \cite{Fri-massive-scalar}, dust \cite{Fri-dust} and, most recently, massless Vlasov matter \cite{Joudioux}.  

Since the CEFE can be written as a symmetric hyperbolic system, well-established results in analysis (e.g. from \cite{Kato}) can be employed to derive local existence, uniqueness and stability results for the Cauchy problem of the CEFE. The strength of the CEFE approach lies in the fact that local results for the CEFE near $\Scri$ imply semi-global results for the corresponding Einstein field equations on $(M,g)$. 

In this section we outline some of the approaches and results that lead to the proof of our stability theorem. For more detailed discussions of various aspects of the CEFE, as well as different formulations, the reader is referred to \cite{FraLR, Fri-review, JVKbook}.  

Since the CEFE are regular everywhere in $(\hat{M}, \hat{g})$ we may consider any spacelike hypersurface ${\hat{\cal{S}}}$ in $(\hat{M}, \hat{g})$ to set up a Cauchy problem for the CEFE. In particular, if  $\Lambda>0$ we may consider the so-called {\em Cauchy problem at infinity} by setting ${\hat{\cal{S}}}={\hat{\cal{S}}}_\infty$, i.e. ${\hat{\cal{S}}} =\Scri^+$. This set up allows us to evolve the CEFE both to the past and future of ${\hat{\cal{S}}}_\infty $.

Let's choose a suitable Bianchi space-time and take the limits from Proposition \ref{prop-constraints at Scri} as initial data for the CEFE on $\hat {\cal S}_\infty$. Then there exists a time interval $[\tau_\infty - \Delta \tau, \tau_\infty + \Delta \tau] $ that leads to a unique regular solution to the CEFE. Since the related physical solution in the past of $\tau_\infty$ is unique, it must be the Bianchi space-time whose initial data at infinity we used. We are in fact able to expand our solution of the  CEFE to the  past to $[0,\tau_\infty+\Delta \tau] \times \hat{{\cal{S}}}$. Hence we have a conformal extension of our Bianchi space-time up to and beyond $\Scri^+$. Such a conformal extension may seem larger than necessary, but we will need it for the stability proofs, where we seek a regular solution that extends beyond $ \Scri^+$ by some $\Delta \tau$. This guarantees that the location of the perturbed conformal boundary remains within our coordinate patch, and hence asymptotic simplicity and geodesic completeness are ensured for the perturbed space-time.

However it is unclear whether such a suitable Bianchi space-time exists. Hence in this section we will address the following question:
\\\\
\textit{Does a Bianchi space-time, with a chosen matter model satisfying Assumptions \ref{Assumptions Wald+} and \ref{Assumption for Pi} with $k\le-1$, possess enough regularity to give rise to a solution of the conformal Einstein field equations (CEFE) up to (and beyond) $ \Scri^+$?}
\\\\
This question will be answered affirmatively in Theorem \ref{space-times-regular-CEFE} for some matter fields. 

\subsection{Bianchi space-times as regular solutions of the CEFE}
\label{Sec:CEFE}
The variables used for the  formulation of the CEFE can be grouped into two sets: 
\begin{enumerate}
\item  Conformal geometric variables, including the unphysical connection and curvatures, 
\item  Matter variables. 
\end{enumerate}
The first set is common to the different formulations of the CEFE, while the second set depends on the underlying matter model. 
Of particular interest to us are the conformal constraints at infinity, which arise when we set $\Omega=0$ in the CEFE. More specifically, the geometric conformal variables need to satisfy \cite{Fri-dS, JVKbook}:
\begin{eqnarray}
\label{constraint d s chi} 
&&  d_0  \circeq -\lambda, \quad 
s\circeq -\lambda\,w, \quad
\hat\chi_{\a\b} \circeq -w h_{\a\b},  \\
\label{constraint Schouten magnetic Weyl}
&& \hat{P}_{0 \b} \circeq D_\b w, \quad
\hat{P}_{\a\b}\circeq \hat{p}_{\a\b}, \quad \mathcal{H}_{\a\b} \circeq {\color{red} \pm}\frac{\hat{Y}_{\a\b}}{\lambda}\\
\label{constraint electric Weyl} 
&& D^\a \mathcal{E}_{\a\b}	\circeq \Xi_\b,
\end{eqnarray} 
where $w$ is a smooth function on $\Scri^+$ and $\Xi_\b$ depends on the matter model.\footnote{We note that $w$ is a free function for tracefree matter models and dust, whereas one must have $w=0$ in the case of the massive scalar field.}

These equations look very similar to the asymptotic identities \eqref{constraints at Scri 0}-\eqref{constraints at Scri 3} established in Proposition \ref{prop-constraints at Scri}. 
Below we show that for some of the Bianchi space-times considered in this article the two sets of equations are indeed equivalent. Moreover, we will conclude that these Bianchi space-times give rise to regular solutions of the associated CEFE up to and beyond conformal infinity.
More precisely, we establish:
\begin{theorem} {\bf (Regular solutions of the CEFE) }
\label{space-times-regular-CEFE}
\\
Suppose we have a Bianchi space-time satisfying Assumption 1, whose matter model corresponds to either to:
\begin{itemize}
\item Einstein-Maxwell fields, 
\item Aligned radiation fluids,
\item Aligned dust fluids,
\item Massive scalar field with potential given by \eqref{scalar-field-potential}, 
\item Massless Vlasov matter, with $f\in C^1$ with compact support and $v$ bounded away from zero initially.
\end{itemize}
Then the following hold:
\begin{enumerate}
\item  The asymptotic identities \eqref{constraints at Scri 0}-\eqref{constraints at Scri 3} at $\tau=\tau_\infty$ are equivalent to the conformal constraints at infinity \eqref{constraint d s chi}-\eqref{constraint electric Weyl} derived from the corresponding conformal Einstein field equations.
\item These Bianchi space-times give rise to regular solutions of the conformal Einstein field equations (CEFE) up to and even beyond $\tau=\tau_\infty$, i.e. conformal infinity.
\end{enumerate}
\end{theorem}
The proof is given in Appendix \ref{proof-thm2} and it is an application of the results of \cite{Fri-EMYM, Fri-massive-scalar, LueVal12b, Fri-dust} and \cite{Joudioux} for Maxwell fields,  massive scalar fields, radiation fluids, aligned dust and massless Vlasov matter, respectively, by checking the conformal constraints at infinity as well as the conditions for regularity of the CEFE for each individual matter field given in those papers. Moreover, in the proof of part 1 we will make use of Proposition \ref{thm-constraints} and for part 2 we will use the estimates derived in Appendix \ref{appendix-vlasov}.

\begin{remark} 
Following the discussion in section \ref{multiple-matter} and in \cite{Lue13}, the above results are expected to also extend to Bianchi space-times containing several of the above matter models, each coupled only to gravity. \end{remark}

What is remarkable is that the asymptotic identities \eqref{constraints at Scri 0}-\eqref{constraints at Scri 3} and the conformal regularity even hold for some matter models for which regular CEFE have not been formulated yet, such as aligned perfect fluids or massive Vlasov matter. Moreover, we observe that the decay rates for the energy-momentum tensor of an aligned perfect fluids depend on $\gamma$. It has been possible to formulate CEFE for dust, which in terms of asymptotics is one of the weaker cases amongst the perfect fluids. 
These observations give further hope that the CEFE can be formulated for a wider class of matter models. But we won't pursue this problem further in this paper.

\subsection{Non-linear stability of solutions to the CEFE}
 \label{Section - stability} 
In the previous section, we established that large classes of Bianchi space-times satisfying Assumption \ref{Assumptions Wald+} give rise to regular solutions of the CEFE up to and beyond the conformal boundary. In the following, we shall use the conformal extension of a Bianchi space-time satisfying Theorem \ref{space-times-regular-CEFE}, and its solution of the CEFE, as our {\em background space-time}. 

Below we describe how the results can be adapted to show that the background Bianchi space-time is future non-linearly stable (with respect to small perturbations within the respective class of matter models). The results will be obtained without further analysis, but as a direct application of the results in \cite{Fri-EMYM, Fri-massive-scalar, Fri-dust, Joudioux, LueVal12b}\footnote{Note the proof for massless Vlasov can not be directly applied to Bianchi space-times and requires some adaptation (see Remark \ref{Vlasov-remark})} , which derive and analyse the CEFE for each of those matter fields and prove the related stability theorems. The proofs follow a common strategy by applying results of \cite{Kato} to the symmetric hyperbolic formulation of the CEFE.

Let $\cal \hat S$ be a suitable 3-dimensional space-like Cauchy surface within our Bianchi space-time. In particular, we may choose our initial surfaces of homogeneity on which the conditions of Assumption \ref{Assumptions Wald+} are satisfied. Denote by $\bm {\mathring  \omega}_0$ sufficiently smooth initial data for the CEFE on $\cal \hat S$. Further, let $(\mathring{M}, \mathring{g})$ denote the associated unphysical background solution to the CEFE on $[0, \tau_\infty + \Delta \tau] \times \hat {{\cal S}}$.
  \footnote{Alternatively, for the stability proof, we can also set $\cal \hat S=\cal \hat S_\infty$ which represents $\Scri^+$ and let $\bm {\mathring  \omega}_0$ denote the Cauchy data at infinity induced by the CEFE.}
Now let $\bm {\omega}_0$ denote new initial data for the CEFE given on $\cal \hat S$ where $\bm {\omega}_0=\bm {\mathring  \omega}_0+ \bm{\breve\omega}_0$ and $\bm{\breve\omega}_0$ corresponds to perturbations which are sufficiently smooth and sufficiently small on a suitable Sobolev norm $H^m$ (with $m\ge 5$) but otherwise arbitrary (i.e. they can be non-homogeneous). Then, the well-posedness of the respective Cauchy problem for the CEFE follows from \cite{Kato}. Furthermore, 
taking sequences of such initial data $\bm \omega^{(n)}_0=\bm {\mathring  \omega}^{(n)}_0+ \bm{\breve\omega}^{(n)}_0$ with sufficiently small perturbations ${\bm{\breve \omega}}^{(n)}_0$, one obtains that the corresponding solutions $\bm {\breve \omega}^{(n)}$ converge to zero uniformly in conformal time. In particular we can choose sufficiently small perturbations so that the perturbed null infinity is located in $\cal \hat S=\cal \hat S_\infty$. It then follows that the perturbed space-time is asymptotically simple and future  geodesically complete. Since we are only considering the evolution of perturbations to the future, we refer to such stability results as semi-global. The details of the proof depend on each matter content and conditions outlined in \cite{Fri-EMYM, Fri-massive-scalar, Fri-dust, Joudioux, LueVal12b}.

We now state our final theorem:
\begin{theorem}{\bf (Stability of Bianchi space-times)}
 \label{Main theorem}
\\
Suppose $(M,g)$ is a Bianchi space-time, whose surface of homogeneity $\hat {\cal{S}}$ is smooth, compact and orientable. Suppose further that the matter content in $(M,g)$ is given either by  
	\begin{itemize}
		\item Einstein-Maxwell field,
		\item Aligned radiation fluid,
		\item Aligned dust fluid,
		\item Massive scalar field, 
	\end{itemize} 
	and satisfies the conditions of Theorem \ref{space-times-regular-CEFE}. 
	
	Then $(M,g)$ is future non-linearly stable to small non-homogeneous perturbations on $\hat {\cal{S}}$ within each class of matter models in the following sense. 
The perturbations of the initial data of the Bianchi space-time are sufficiently small in some suitable Sobolev norms $H^m$ with $m\ge 5$. 
It then follows that these perturbations decay to zero in time, the perturbed space-time approaches de Sitter at late times and is future geodesically complete.

\end{theorem}
\begin{proof}
The proof follows from Theorem \ref{prop-future-asymp-simple}, Theorem \ref{space-times-regular-CEFE} and, as discussed above, from applying of the stability results of \cite{Fri-EMYM} for the Einstein-Maxwell system, \cite{Fri-dust} for dust fluids, \cite{Fri-massive-scalar} for massive scalar fields and of \cite{LueVal12b} for radiation fluids.  
\end{proof}

\begin{remark} 
\label{Vlasov-remark}
The stability proof for massless Vlasov matter in \cite{Joudioux} is specifically set up for Minkowski and de Sitter space using the cylindrical coordinates on the Einstein cylinder. The work in Appendix \ref{appendix-vlasov} derives the relevant boundedness in a more general setting. Using the patching argument detailed in \cite{Fri-EMYM} would allow to adapt the stability proof of \cite{Joudioux} for perturbations of the initial data on $\cal \hat S$ or $\cal \hat S_\infty$  and show that a unique massless Vlasov space-time exists which stays close to the original Bianchi space-time. The remaining steps follow the approach of the  references used above.

\end{remark}
Our result extends the cosmic no-hair result of Wald \cite{Wal83} in the sense that the almost spatially homogeneous space-times of Theorem \ref{Main theorem}, will locally approach de Sitter space-time at late times. 
 However, it is  interesting to note that the underlying conformal evolution preserves some information of the Bianchi type of the original unphysical background space-time via $\hat{\Gamma}\tensor{}{c}{ab} $. We called this property {\em conformal hair}. We point out that this property is potentially interesting for cosmologies based on conformal geometry such Penrose's cyclic models \cite{Penrose-CCC,Tod-CCC}.
\subsection*{Acknowledgments}
The authors thank Helmut Friedrich, J\'er\'emie Joudioux, Ernesto Nungesser, Juan Valiente-Kroon  and Max Thaller for useful discussions. CL would like to thank  for the hospitality of Universidade do Minho, where this work was started. FM thanks hospitality from EPFL where this work was concluded, CMAT, Univ. Minho, through FCT projects UIDB/00013/2020 and UIDP/00013/2020 and FEDER Funds COMPETE, as well as CAMGSD, IST-ID, projects UIDB/04459/2020 and UIDP/04459/2020.
\appendix
\section{Orthonormal tetrad formalism}
\label{sec1+3}
In this appendix, we briefly review the orthonormal tetrad formalism in General
Relativity (see e.g.  \cite{EllMac69}, \cite{MacCallum73} and \cite{vElUgg96}
for more details).

We identify the vector field ${\bf u}$ with a time-like unit basis field
$\bm{e}_0$ and take an orthonormal spatial basis ${\bm e_\a}$. The orthonormal tetrad field
$\bm{e}_a$ is such that $\bm e_0\cdot \bm e_0=-1$, $\bm e_\a \cdot \bm e_\b =\delta_{\a \b}$ and 
$\bm e_0\cdot \bm e_\a=0$. 
The covariant derivative of $\bm e_a$ along $\bm e_b$ is given by 
\begin{equation}
\label{ola13}
\nabla_{\bm e_b} \bm e_a= \Gamma^c{}_{ab} \bm e_c,
\end{equation}
where $\Gamma^c_{~ab}$ are the so-called Ricci rotation coefficients or frame connection coefficients. The commutators associated with the tetrad 
$\{\bm e_a\}$ can be written as
\begin{equation}
\label{com}
[\bm e_a,\bm e_b]=\nabla_{\bm e_a} \bm e_b - \nabla_{\bm e_b} \bm e_a= \gamma^c{}_{ab}\bm e_c,
\end{equation}
where $\gamma^c_{~ab}=\gamma^c_{~[ab]}$ 
and the geometrical objects $\Gamma^a{}_{bc}$ and $\gamma^c{}_{ab}$, in the case 
of zero torsion are related by
\begin{eqnarray}
\label{ola14}
\gamma^c{}_{ab}&=&\Gamma^c{}_{ba}-\Gamma^c{}_{ab}\\
\label{olala}
\Gamma_{abc}&=&\half(\eta_{ad}\gamma^d{}_{cb}+\eta_{bd}\gamma^d{}_{ac}-
\eta_{cd}\gamma^d{}_{ba}),
\end{eqnarray}
where $\Gamma_{abc}=g_{ad}\Gamma^d_{~bc}$. 
In order to calculate the expressions for the commutators (\ref{com}) one first has
to calculate the Ricci rotation coeficients (\ref{olala}).
Since one has $\Gamma_{abc}=-\Gamma_{cba}$, then $\Gamma_{abc}$ has 24 independent components.  
Using the formalism of Section \ref{setup}, we obtain   
\begin{eqnarray}
\label{ola15}
\Gamma_{\a00}&=& A_\a
\\
\label{olapa}
\Gamma_{\a\b0}&=&\sigma_{\a\b}+ \Hubble\delta_{\a\b}+\w_{\a\b}.
\end{eqnarray}
It is also useful to define
the quantity $\Omega^a$ as 
\begin{equation}
\label{Omega_definition}
\Omega^a=\half\eta^{abcd}(\nabla_{\bm e_0}{\bm e}_c) \bm e_d u_b,
\end{equation}
which can be interpreted as the angular velocity of
the triad $\bm e_\a$ 
in the rest space of an observer with four-velocity $\bm e_0$, as the observer propagates $\bm e_\a$
along $\bm e_0$. It is therefore not part of the space-time dynamics and with a tetrad choice one 
can always make $\Omega_a=0$, in which case, one says that the frame is Fermi propagated along ${\bf u}$. 
One can show using (\ref{ola13}), (\ref{olala}) that
\begin{equation}
\label{ola17}
\Gamma_{\a 0\b}=\varepsilon_{\a\b\gamma}\Omega^\gamma.
\end{equation}
Now, the commutation functions $\gamma^\a{}_{\b \g}$ can be decomposed
as 
\begin{eqnarray}
\label{ola18}
\g^\a{}_{\b\g}=\e_{\b \g \delta} n^{\delta \a}+\delta^\a_{~\g} a_\b-\delta^\a_{~\b} a_\g
\end{eqnarray}
where $n^{\a \b}=n^{(\a \b)}$ and  
the objects $\varepsilon_{\b \g \delta}$ and $\delta^\a_{~\b}$ are the Levi--Civita and
the Kronecker delta symbols respectively. Then, the remaining 
spatial Ricci rotation coefficients can be written as
\begin{equation}
\label{ola19}
\Gamma_{\a\b\gamma}=2a_{[\a}\delta_{\gamma]\b}+\varepsilon_{\b\delta[\a}n^\delta_{\gamma]}+
\frac{1}{2}\varepsilon_{\a\gamma\delta}n^\delta_\b.
\end{equation}
The 16 Jacobi identities 
\begin{equation}
\label{jacobi}
[[\bm e_a,\bm e_b],\bm e_c]+[[\bm e_b,\bm e_c],\bm e_a]+[[\bm e_c,\bm e_a],\bm e_b]=0,
\end{equation}
can be written as
\begin{equation}
\label{ola22}
\bm e_{[a}(\gamma^d{}_{bc]})+\gamma^e{}_{[ab}\gamma^d{}_{c]e}=0
\end{equation}
and are equivalent to the 16 Bianchi identities $R_{a[bcd]}=0$. 
These identities will provide evolution and constraint
equations for the commutation functions, as we shall see in Appendix \ref{appendix-equations}.

The Riemann curvature can be decomposed in terms of the Weyl tensor $C\tensor{ab}{c}{d} $ and the Schouten tensor 
$P_{ab} = \tfrac{1}{2} R_{ab} - \tfrac{1}{12} R g_{ab} $ 
as
\begin{equation}
\label{Riemann curvature decomposition}
R^{ab}{}_{cd}
= C^{ab}{}_{cd} + 2\delta^{[a}{}_{[c}R^{b]}{}_{d]} - 
\frac{1}{3}R\delta^{a}{}_{[c}\delta^{c}{}_{d]}
= C^{ab}{}_{cd} + 4\delta^{[a}{}_{[c}P^{b]}{}_{d]}.
\end{equation}
In terms of the energy-momentum tensor, the Schouten tensor takes the form
\begin{equation}
\label{Schouten in terms of matter}
P_{ab}=\frac{1}{2}T_{ab} - \frac{1}{6}Tg_{ab}+\frac{1}{2}\lambda^2 g_{ab}
=\frac{1}{2}(\rho+p)u_a u_b + \frac{1}{6}\rho g_{ab} + q_{(a}u_{b)} +\frac{1}{2} \pi_{ab} +\frac{1}{2}\lambda^2 g_{ab}.
\end{equation}
The Weyl tensor can be decomposed into its "electric part" $E_{ab}$ and 
"magnetic part" $H_{ab}$ relative to the vector field
 ${\bf u}$ as
\begin{equation}
\label{ola23}
E_{ac}=C_{abcd}u^b u^d,~~~H_{ac}=\frac{1}{2}\eta_{ab}^{~~st} C_{stcd} u^b u^d.
\end{equation}
The tensors $E_{ab}$ and $H_{ab}$ satisfy $E_{ab}u^b=H_{ab}u^b=0$ and one can show that
$C_{abcd}=0\Leftrightarrow E_{ab}=H_{ab}=0.$
 In Appendix \ref{appendix-equations}, we will write the expressions for $E_{ab}$ and $H_{ab}$ with respect to an orthonormal frame, for the case where $A^\a=\omega^\a=0$ and $\Omega^\alpha=0$.
%
\section{Evolution and constraint equations in the physical system}
\label{appendix-equations}
Below we give the reduced Einstein field equations under the assumptions of spatial homogeneity and 
$A^\a=\omega^\a=0$ and $\Omega^\alpha=0$, which correspond to our gauge choice in the main text:
\begin{eqnarray}
\label{jac4SH}
{\bm e}_0 (a^\a)&=& -\Hubble a^\alpha - a_\beta \sigma^{\alpha\beta} \quad\quad
\\
\label{jac5SH}
{\bm e}_0 (n^{\a\b})&=& -\Hubble n^{\alpha\beta} - 2\sigma\tensor{}{(\alpha}{\gamma}n^{\beta)\gamma} \quad\quad 
\\
\label{efe0SH}
{\bm e}_0(\Hubble)&=&-\Hubble^2-\frac{2}{3}\sigma^2 - \frac{1}{6}(\rho+3p)+\lambda^2
\\
\label{efe2SH}
{\bf e}_0(\sigma^{\a\b})&=&-3 \Hubble \sigma^{\a\b} + \pi^{\a\b} -~^{(3)}\hspace{-0.05cm}S^{\a\b}\quad\quad 
\\
\label{ids1SH}
{\bf e}_0({\rho})&=&- 3(\rho+p)\Hubble-\pi_{\a\b}\sigma^{\a\b} + 2 a^\a q_\a \\
\label{efe51SH}
{\bf e}_0 (q_\a)&=&
-4\Hubble q_\a - \sigma_{\a\b}q^\b + 3a_\b\pi^\b_\a+\pi^\gamma_\b\varepsilon_{\gamma\a\delta}n^{\b\delta},
\end{eqnarray}
where 
\begin{eqnarray}
\label{efe70SH}
^{(3)}\hspace{-0.05cm}S_{\a\b} &:=& 2n_{\a\gamma}n^{\gamma}_{~\beta}-n^{\gamma}_{~\gamma}n_{\a\b}-\frac{1}{3}\delta_{\a\b} [ 2n^{\a\b}n_{\a\b}-(n^\a_\a)^2 ] + 2 \varepsilon^{\gamma\delta}_{~~(\a}a_{|\gamma|} n_{\b)\delta} \\
\label{efe7SH}
^{(3)}\hspace{-0.05cm}R&=& -6a_\a a^\a-n^{\a\b}n_{\a\b}+\frac{1}{2}(n^\a_{~\a})^2\\
\label{efe1SH}
q^\a &=& 3\sigma^{\b\a} a_\b - \varepsilon^{\a\b\gamma} n_{\gamma\delta}\sigma_\b^\delta
\end{eqnarray}
and $^{(3)}\hspace{-0.05cm}S_{\a\b}$ and $^{(3)}\hspace{-0.05cm}R$ are the trace-free and trace parts of the intrinsic 3-Ricci curvature $^{(3)}\hspace{-0.05cm}R_{\a\b}$ of the spacelike surfaces.\footnote{This would not be the case if $\omega^{\alpha}\ne 0$,  as $^{(3)}\hspace{-0.05cm}R_{\a\b}$ would no longer be a tensor, see \cite{vElUgg96} for more details.}

From the contraction of (\ref{efe0SH}) and substitution of (\ref{efe2SH}), one gets  
\begin{equation}
\label{efe6SH}
3\Hubble^2=\sigma^2+\rho+3\lambda^2 - \frac{^{(3)}\hspace{-0.05cm}R}{2}.
\end{equation}
In addition, using \eqref{ola23}, one can write the expressions for $E_{\a\b}$ and $H_{\a\b}$ with respect to an orthonormal frame as
\begin{eqnarray}
\label{ola26SH}
E_{\a\b}&=&-\half\pi_{\a\b}+\Hubble\s_{\a\b}-\s_{\a\g}\s^\g_\b +\frac{2}{3} \delta_{\a\b}\s^2
+^{(3)}\hspace{-0.1cm}S_{\a\b}\\
\label{ola28SH}
H_{\a\b}&=&\half n^\g_\g\s_{\a\b} - 3 n_{(\a}^{\g}\s_{\b)\g} + \delta_{\a\b} n_{\g\delta}\s^{\g\delta}
- \varepsilon^{\g\delta}{}_{\a}a_{|\g|} \s_{\b)\delta}.
\end{eqnarray}
For the homogeneous spacelike hypersurfaces, the 3-dimensional Cotton-York tensor is given by
\begin{equation}
\label{3-Cotton-York}
Y_{\a\b} = \half n^\g_\g \,^{(3)}\hspace{-0.05cm}S_{\a\b} - 3 n_{(\a}^{\g} \,^{(3)}\hspace{-0.05cm}S_{\b)\g} + \delta_{\a\b} n_{\g\delta} \,^{(3)}\hspace{-0.05cm}S^{\g\delta}
- \varepsilon^{\g\delta}{}_{\a}a_{|\g|} \,^{(3)}\hspace{-0.05cm}S_{\b)\delta}
\end{equation}
Now, we may use \eqref{ola26SH} to rewrite \eqref{efe2SH} as
\begin{equation}
\label{efe2SHrewritten}
{\bf e}_0(\sigma^{\a\b})=-2\Hubble \sigma^{\a\b} -\s_{\a\g}\s^\g_\b +\frac{2}{3} \delta_{\a\b}\s^2 - E_{\a\b} +\frac{1}{2} \pi^{\a\b}. 
\end{equation}
In turn, the evolution equations for the electric and magnetic parts of the Weyl tensor are:
\begin{eqnarray}
\label{EpropSH}
{\bf e}_0(E^{\alpha\beta}+\frac{1}{2}\pi^{\alpha\beta})&=& - 3\Hubble (E^{\alpha\beta} + \frac{1}{6}\pi^{\a\b}) + 3 \sigma^{(\alpha}_{~~\gamma}(E^{\beta)\gamma}-\frac{1}{6}\pi^{\beta)\gamma})-\delta^{\alpha\beta} \sigma_{\gamma\delta}(E^{\gamma\delta}-\frac{1}{6}\pi^{\gamma\delta})\nonumber\\
&&-\varepsilon^{\gamma\delta(\alpha}a_\gamma H^{\beta)}_{~\delta}+
\frac{1}{2}n^{\gamma}_{~\gamma}H^{\alpha\beta}-3n^{(\alpha}_{~~\gamma}H^{\beta)\gamma}+\delta^{\alpha\beta}n_{\gamma\delta}H^{\gamma\delta}\nonumber\\
& &-\half(\rho + p)\sigma^{\a\b} -\frac{1}{2} a^{(\alpha}q^{\beta)}+\frac{1}{6}\delta^{\alpha\beta} a_\gamma q^\gamma
  +\varepsilon^{\gamma\delta(\alpha}\frac{1}{2} n^{\beta)}_{~~\gamma}q_\delta
  \\
\label{HpropSH}
{\bf e}_0({H}^{\a\b}) &=& - 3\Hubble H^{\a\b} + 3 \sigma^{(\alpha}_{~~\gamma}H^{\beta)\gamma} -\delta^{\a\b}\sigma_{\gamma\delta}H^{\gamma\delta} + \varepsilon^{\gamma\delta(\a} 2\Omega_\gamma H^{\beta)}_{~~\delta}\nonumber\\
&& +3 n^{(\alpha}_{~~\gamma}(E^{\beta)\gamma}-\frac{1}{2}\pi^{\beta)\gamma})-\frac{1}{2}n^\gamma_{~\gamma}(E^{\a\b}-\frac{1}{2}\pi^{\a\b})-\delta^{\a\b}n_{\gamma\delta}(E^{\gamma\delta}-\frac{1}{2}\pi^{\gamma\delta})\nonumber\\
&&+\varepsilon^{\gamma\delta(\alpha} a_\gamma(E^{\beta)}_{~~\delta}-\frac{1}{2}\pi^{\beta)}_{~~\delta}) + \varepsilon^{\gamma\delta(\alpha} \sigma^{\beta)}_{~~\gamma} q_\delta.
\end{eqnarray} 
\section{Estimates for the conformal Einstein-Vlasov system}
\label{appendix-vlasov}
The stability of de Sitter space-time as solution to the massless Einstein-Vlasov system was proved in \cite{Joudioux} using conformal techniques and Kato's theorem. In this paper, instead of de Sitter, we have spatially homogeneous backgrounds. However, some proofs of \cite{Joudioux} can be adapted to our case, provided we have enough control on the relevant variables for the system. In this appendix, we give further details about how this can be done.

In our frame coordinates, the Vlasov energy-momentum tensor \eqref{Einstein-Vlasov} is given 
\begin{equation}
\label{Tab-Vlasov}
\hat T_{ab}=-\int_{{\cal P}} \hat f(x^\mu ,v_a) \frac{v_a v_b}{v} dv_1dv_2dv_3,
\end{equation}
where $v_a$ are the unphysical ${\bf \hat e}_a$ frame components of $p_\mu$ so that $\hat g^{0\mu}p_\mu=-v_0=v:=\sqrt{v_1^2+v_2^2+v_3^2}$. Note, as  we only work with $v_a$ in the unphysical setting we omit the hat, which would otherwise clutter the equations below. $\hat f$ is the density function $f$ expressed in terms of the unphysical variables $v_a$. 

The symmetric hyperbolic PDE system in \cite{Joudioux} involves geometric variables as well as the matter variables composed of
\begin{equation}
\label{vlasov-matter-variables}
v_a,~ \hat f,~ {\bf \hat e}_a(\hat f)~~\text{and}~~\partial_{ v_a} \hat f.
\end{equation}
To apply Kato's theorem to our backgrounds, we first must ensure that those variables are bounded in the background and, furthermore, that the $v_a$ variables are bounded away from zero. We start by proving this last condition which ensures that some matrix coefficients of the symmetric hyperbolic system are non-zero. In order to ensure that, we follow the strategy of \cite{Joudioux} and show, for our background, that if $v_a$ is bounded away from zero initially then it will remain so for a sufficient amount of time. 

Having in mind the application of Theorem 5.1 of \cite{Joudioux} to our case, we follow closely their notation and choice of parameters.
Denoting the frame metric by $\eta^{ab}$ we start by writing the massless Vlasov equation \eqref{Vlasov-equation} in frame coordinates as
\begin{equation}
\label{Vlasov-frame}
\eta^{ab} v_a \hat e_b^\mu\partial_{x^\mu} \hat f+\eta^{ab} v_av_c \hat \Gamma_{a~d}^{~c} \partial_{v_d} \hat f=0
\end{equation}
and define the set
$$
\Omega_\delta=\{(v_0,v_1,v_2,v_3)\in \R^4: v_0=-v, ~\delta\le v \le \frac{1}{\delta} \}.
$$ 
Now, consider the characteristic system of the massless Vlasov  equation \eqref{Vlasov-frame} formulated in terms of the unphysical variables in $(\hat{M},\hat{g})$ as
\begin{eqnarray}
\label{xdot}
\dot x^\mu (s)&=&\eta^{ab} \hat e_a^\mu(x(s)) v_b(s)\\
\label{vdot}
\dot v_d(s)&=&\eta^{ab} \hat \Gamma_{b~d}^{~c} v_a(s) v_c(s)
\end{eqnarray}
with initial conditions such that $v_a(0)\in \Omega_{1/2}$ and  $x^0(0)=0$. We follow the analysis of \cite{Joudioux} and 
consider $x^0(s)\in [0,\tfrac{3}{2}\tau_\infty)$. The aim is to show that $v$ does not vanish on $[0,\tfrac{3}{2}\tau_\infty)$, and therefore neither on $[0,\tau_\infty]$. 

Observe that $\hat \Gamma_{b~0}^{~c}= \hat{\sigma}_{b}^{~c}$ and let $\hat\sigma_{M}=\displaystyle{\sup_{\tau, \hat{\epsilon}_a} |\hat \sigma(\hat{\epsilon}_a \hat{\epsilon}_b)(\tau)|>0}$ where $\hat{\epsilon}_a$ is any triad of space-like unit vectors orthogonal to ${\bf \hat e}_0$. Since the conformal shear $\hat \sigma_{ab}$ is bounded, $\hat \sigma_{M}$ must be finite.
The equation for $v_0=-v$ in \eqref{vdot} gives
$$|\dot v(s)|=|\dot v_0(s)|=|\eta^{ab} \hat \Gamma_{b~0}^{~c} v_a(s) v_c(s)| \le \hat \sigma_{M} v^2.$$
Upon integration we get
\begin{equation}
\label{vnoteq}
|v(s)|\ge \frac{1}{\frac{1}{v_*}+\hat\sigma_{M}s}\ge \frac{1}{2+\hat\sigma_{M}s}
\end{equation}
as well as 
\begin{equation}
\label{upperbound}
|v(s)|\le \frac{1}{\frac{1}{v_*}-\hat\sigma_{M}s}\le \frac{1}{\frac{1}{2}-\hat\sigma_{M}s},
\end{equation}
where we have used $1/2\le v_* \le 2$. It thus follows that $v$ does not vanish as long as $s$ is finite. Thus  the next step is to show that $s$ is indeed finite as long as $x^0(s)\in  [0,\tfrac{3}{2}\tau_\infty)$. 
 
Now, \eqref{xdot} and \eqref{vnoteq} give $|\dot x^0 (s)| \ge |v(s)| \ge (2+\hat\sigma_{M}s)^{-1}$, which integrated gives
$$
x^0(s)\ge \frac{1}{\hat\sigma_M}\ln{(1+\tfrac{1}{2}\hat\sigma_M s)}.
$$
It follows that $x^0(s)$ reaches $\tfrac{3}{2}\tau_\infty$ before $s$ reaches
$$
s_M:= \frac{2}{\hat \sigma_M}\left( e^{(\tfrac{3}{2}\hat\sigma_M \tau_\infty)}-1\right).  
$$
In particular there exists $s_\infty < s_M$ such that $x^0(s_\infty)=\tau_\infty$.
Finally, substituting $s_M$ in \eqref{vnoteq} we find that for all $s\le s_M$
$$
 v \ge e^{(\tfrac{3}{2}\hat\sigma_M \tau_\infty)}.
$$
Thus in particular $v$ is bounded away from 0 up to $\Scri^+$, as desired.

Recall that the conformal shear $\hat\sigma_{ab}$ decays to zero at $\Scri^+$. Thus we can select $s_1 \in [0,s_\infty)$ such that the bound $\hat\sigma_{M}$ becomes sufficiently small to guarantee $\hat\sigma_{M}s < \frac{1}{4}$ for all $s>s_1$. Then \eqref{vnoteq} and  \eqref{upperbound} imply that $v$ is bounded above and below at $\Scri^+$, in fact
$$
v_a(0)\in \Omega_{1/2} \implies v_a(s) \in \Omega_{1/4} \quad \forall s \in [s_1,s_\infty].
$$
As highlighted in  \cite{Joudioux} this implies that the physical equivalent of $v_a$ is bounded by $O(e^{-\lambda t})$. We note that choice of $\delta =1/2$ is for illustrative purpose of the derivation. Other values may also be suitable.

In the following we assume that $f(t_*)$, and hence $\hat f (\tau_*)$, is spatially homogeneous.  Moreover we assume that $f(t_*, v)$ is a $C^1$ function of compact support on $ {\cal P}$ and $v_a(0)\in \Omega_{\delta} $ for some  suitable $\delta$. The last conditions implies that $v_a$ is bounded and has compact support in some $\Omega_{\delta*} \subset {\cal P}$ up to $\Scri^+$:
Given these assumptions and our gauge choices we can deduce that up to $\Scri^+$
\begin{itemize}
\item $v_a$ is bounded and has compact support in some $\Omega_{\delta*} \subset {\cal P}$,
\item $\partial_{ v_a} \hat f$ are continuous and bounded, 
\item ${\bf \hat e}_i(\hat f)=0$ for $i=1,2,3$, and hence
\item ${\bf \hat e}_0(\hat f)$ is bounded by \eqref{Vlasov-frame}.
\end{itemize} 

\section{Proof of Theorem \ref{space-times-regular-CEFE}}
 \label{proof-thm2}
\subsection*{Proof of Part 1}
The equations \eqref{constraints at Scri 0}-\eqref{constraints at Scri 2} agree with \eqref{constraint d s chi}-\eqref{constraint Schouten magnetic Weyl} if one sets $w=0$. Thus, the only equation that requires closer attention is \eqref{constraint electric Weyl}, which depends on the matter model.

For \textit{tracefree matter}, we get $\Xi_\b = - \lambda \,\Omega^{-4}\,q_{\b}$, where $q_\b = q_\mu e^\mu_\b $. Hence, \eqref{constraints at Scri 3} and \eqref{constraint electric Weyl} agree directly.

For \textit{massive scalar fields}, we have $q_\a=0$. If the scalar field satisfies the assumptions given in Section \ref{Section - massive scalar field}, then Proposition \ref{prop-constraints at Scri} holds and \eqref{constraints at Scri 3} reads $D^\a \mathcal{E}_{\a\b} \circeq 0$. The corresponding CEFE constraint \cite{Fri-massive-scalar}
uses
$ \Xi_\b \circeq \frac{1}{3} \lambda (\psi_0 D_\b \psi_1 - 2 \psi_1 D_\b \psi_0)$,
where $\psi_0$ and $\psi_1$ are smooth functions on $\Scri^+$ representing $\hat \phi$ and ${\bf \hat e}_0(\hat\phi)$, respectively. Moreover, if $X^\a$ is a conformal vector field admitted by the induced metric $h$ on $\Scri^+$, then it must satisfy $\int_{\Scri^+ } X^\a \Xi_\a d\mu_h=0$, where $d\mu_h $ is the volume element of the metric induced on $\Scri^+$. For a Bianchi space-time with a spatially homogeneous massive scalar field $\phi$ we have  ${\bf e}_\a(\psi_0)=0$ and ${\bf e}_\a(\psi_1)=0$ and, hence, $\Xi_\a \circeq 0$ so that the integral vanishes identically as required.

For \textit{aligned dust}, we have $q_\b=0$ so that \eqref{constraints at Scri 3} reads $D^\a \mathcal{E}_{\a\b} \circeq 0$ once more. The constraint equation at infinity for dust \cite{Fri-dust} has  $\Xi_\b=1/3 D_\b \rho - \rho f_\b$, where $f_\b$ is the $(g_{ab}+u_a u_b)$-projected component of the 1-form $b_a$ associated with the conformal geodesic parametrisation of the flow lines. In addition, the flow lines must meet $\Scri^+$ orthogonally.
The latter condition holds since the fluid is aligned with the normal vector of the surfaces of homogeneity.
For the first condition, spatial homogeneity implies $D_\b \rho=0$ while, for the second term, we use the fact that the flow lines of a perfect fluid in a Bianchi space-time can be reparametrised as conformal geodesics with $b_a=\alpha u_a$, where $\a$ is a funtion of $t$, respectively $\tau$, and $u_a$ is the physical fluid flow velocity \cite{LueTod-Bianchi}. Hence, $ f_\b$ vanishes along the flow lines of an aligned perfect fluid Bianchi space-time and, as a result, $\Xi_\b$ vanishes at $\Scri^+$.

\subsection*{Proof of Part 2}
 
We will show that, for each individual matter model below, the matter related conformal variables are regular and, thus, we will be able to conclude, using the results of \cite{Fri-EMYM,LueVal12b,Fri-dust,Fri-massive-scalar,Joudioux} that the Bianchi space-times under consideration give rise to regular solutions of the CEFE.
Given that the conformal constraints \eqref{constraint d s chi}-\eqref{constraint electric Weyl} are satisfied on ${\hat{\cal{S}}}_\infty$, we can prove the existence of a solution of the CEFE in a neighbourhood to the future of ${\hat{\cal{S}}}_\infty$. In other words from the perspective of the related Bianchi space-time we have a conformal extension beyond $\Scri^+$.
\subsubsection*{Einstein-Maxwell fields}
The CEFE for an Einstein-Maxwell system were derived in \cite{Fri-EMYM} and later reformulated in a gauge using conformal curves \cite{LueVal12a}.
For the CEFE, we need to analyse the ${\bf \hat e}_a $ frame components of the Faraday tensor $F_{\mu\nu} $ and its derivatives $\hat{\nabla}_\lambda F_{\mu\nu}$.
Observe that, due to our convention, the frame components in the physical and unphysical frame are related by
$\hat{F}_{ab}=L^2 F_{ab}$. It follows that it is sufficient to show finiteness of the electric and magnetic field components ${\cal{E}}_{a} $ and ${\cal{H}}_{a} $, as well as their covariant derivatives $\hat{\nabla}_{a} {\cal{E}}_{b} $ and $\hat{\nabla}_{a} {\cal{H}}_{b} $. This was established in Section \ref{EMYM section}.

\subsubsection*{Aligned radiation fluid}

The CEFE for radiation fluids were derived and analysed in \cite{LueVal12b}. The matter-related variables are: the rescaled fluid velocity $\hat{u}^\mu =\Omega^{-1} u^\mu$, the rescaled density $\hat{\rho}=\Omega^{-4}\rho$ as measured by a comoving observer and the ${\bf \hat e}_a $ frame components of their first derivatives $\hat{\nabla}_{\mu} \hat{u}_{\nu} $ and $ \hat{\nabla}_{\mu} \hat{\rho}$.
When the radiation fluid is aligned with the canonical observer, then ${\bf \hat e}_0^\mu = \hat{u}^\mu$. We observe that we have shown earlier that the rescaled density $ \hat{\rho}=L^4\rho$, as well as ${\bf \hat e}_0 (\hat{\rho})$ and $\hat{\nabla}_{a} \hat{u}_{b} = \hat{\Gamma}\tensor{a}{c}{0}\hat{g}_{bc}$ are all finite.

\subsubsection*{Aligned dust}

The CEFE for dust with $\Lambda>0$ were derived in \cite{Fri-dust}. The matter-related variables are the rescaled fluid velocity $\hat{u}^\mu =\Omega^{-1} u^\mu $ and the rescaled density $\hat{\rho}=\Omega^{-3}\rho$, as measured by a comoving observer.
For aligned dust, we have $\hat{u}^\mu = {\bf \hat e}_0^\mu$ once more, and $\hat{\rho}$ agrees with the established variable.
The fluid flow lines of dust space-times can be reparametrised as conformal geodesics. A condition for conformal regularity is that the conformal geodesics must meet the conformal boundary at right angle \cite{Fri-dust}. As was shown in \cite{LueTod-Bianchi}, the trajectories of the canonical observer in a perfect fluid Bianchi space-time can be parametrised as conformal geodesics, whose 1-form has no spatial components. Hence, aligned dust flows along conformal geodesics.

\subsubsection*{Massive scalar field}

The CEFE for the massive scalar field $\phi$ with $\Lambda>0$ were analysed in \cite{Fri-massive-scalar}.
The matter-related variables are the rescaled scalar field $\psi=\Omega^{-1}\phi$ and the ${\bf \hat e}_a $ frame components of its gradient $\hat{\nabla}_a \psi $.
The potential $V(\phi)$ in \cite{Fri-massive-scalar} has the form $V(\phi) = \frac{1}{2} m^2 \phi^2 + \phi^4 U(\phi)$, which agrees with the potential in Section \ref{Section - massive scalar field}, as $m^2=2\lambda^2$, and we have shown that both $\psi$ and ${\bf \hat e}_0(\psi)$ have finite limits at $\tau_\infty$.

\subsubsection*{Massless Vlasov matter}

The CEFE for the massless Einstein-Vlasov system were derived and analysed in \cite{Joudioux}.
In Appendix \ref{appendix-vlasov} we derive several results for Bianchi space-times with massless Vlasov matter based on the assumptions that $f(t_*)$, and hence $\hat f (\tau_*)$, is spatially homogeneous. Moreover we assume that $f(t_*, v)$ is a $C^1$ function of compact support on $ {\cal P}$ and $v_a(\tau_*)\in \Omega_{\delta} $ for some  suitable $\delta$ (as in \cite{Joudioux}), see Appendix \ref{appendix-vlasov} for details. 
Using the results established at the end of Appendix \ref{appendix-vlasov}, our assumptions and our gauge choices imply for the Bianchi space-times considered in this article:
\begin{itemize}
\item The matter variables defined in Eq.(3.2) of \cite{Joudioux} are all bounded up to $\Scri^+$,
\item The momenta $v_a$ are bounded and have compact support in some $\Omega_{\delta*} \subset {\cal P}$.
\end{itemize}
It follows that the all the variables required in the CEFE for massless Vlasov matter are regular up to $\Scri^+$.
\\\\
The above matter types all satisfy Assumption \ref{Assumption for Pi}  automatically with $k=-1$. 
Then, using the results of \cite{Fri-EMYM,LueVal12b,Fri-dust,Fri-massive-scalar,Joudioux}, we obtain a regular solution of the respective CEFE up to $\tau=\tau_\infty$, that is up to conformal infinity, where the Bianchi space-time satisfies the corresponding conformal constraints at infinity according to Part 1. Thus we have suitable initial data at infinity which leads to a regular solution of the CEFE on a time interval $[\tau_\infty - \Delta \tau, \tau_\infty + \Delta \tau] $. Due to uniqueness, the physical space-time to the past of $\tau_\infty$ must be the Bianchi space-time considered in the first place.
 

\end{document}